\newif\iflong
 \longtrue

\newif\ifold

\documentclass[USenglish, a4paper, thm-restate,numberwithinsect, cleveref]{lipics-v2021}

\nolinenumbers
\hideLIPIcs

\usepackage{todonotes}

\usepackage{graphicx} \usepackage{xcolor}

\definecolor{pcol}{rgb}{0, 0.75, .8}
\definecolor{burntorange}{rgb}{0.93, 0.53, 0.18}
\definecolor{chocolate}{rgb}{0.82, 0.41, 0.12}
\definecolor{darkorange}{rgb}{1.0, 0.55, 0.0}

\usepackage{tikz}
\usetikzlibrary{decorations.pathreplacing, decorations.pathmorphing, arrows.meta, calc, shapes.arrows, positioning, matrix}
\usepackage{tkz-graph}
\tikzstyle{vertex}=[circle, draw, inner sep=0pt, minimum size=6pt]

\usepackage{amsthm}
\usepackage{amsmath}
\usepackage{amsfonts}
\usepackage{xspace}
\usepackage{hyperref}
\usepackage{tabularx}
\usepackage[inline]{enumitem}

\newcommand{\dd}{\ensuremath{\mathbf{d}}\xspace}
\newcommand{\cp}{$C_4$-pivotable\xspace}

\newcommand{\Oh}{\mathcal{O}\xspace}

\newcommand{\ie}{i.\,e.,\xspace}
\newcommand{\G}{\ensuremath{\mathcal{G}}\xspace}

\newcommand{\prob}[3]{\begin{quote}  \textsc{#1}\\  \textbf{Input:} #2\\  \textbf{Question:} #3\end{quote}}

\crefname{claim}{Claim}{Claims}

\keywords{temporal paths, gossiping, (multi)graphical degree sequences, edge-disjoint spanning trees, linear time algorithms}

\author{Arnaud Casteigts}{Department of Computer Science, University of Geneva, Switzerland\\ LaBRI, Université de Bordeaux, France}{arnaud.casteigts@unige.ch}{https://orcid.org/0000-0002-7819-7013}{Supported by the French ANR, project ANR-22-CE48-0001 (TEMPOGRAL).}

\author{Michelle {D\"oring}}{
        Hasso Plattner Institute, University of Potsdam, Germany
                \and \url{https://hpi.de/friedrich/people/michelle-doering.html} }{michelle.doering@hpi.de}{https://orcid.org/0000-0001-7737-3903}{German Federal Ministry for Education and Research (BMBF) through the project ``KI Servicezentrum Berlin Brandenburg'' (01IS22092)}

\author{Nils Morawietz}{Institute of Computer Science, Friedrich Schiller University Jena, Germany\\ LaBRI, Université de Bordeaux, France}{nils.morawietz@uni-jena.de}{https://orcid.org/0000-0002-7283-4982}{Supported by the French ANR, project ANR-22-CE48-0001 (TEMPOGRAL).}

\authorrunning{A.~Casteigts,  M.~{D\"oring}, and N.~Morawietz}

\Copyright{Arnaud Casteigts, Michelle Döring, Nils Morawietz}

\ccsdesc[500]{Theory of computation~Graph algorithms analysis}
\ccsdesc[500]{Mathematics of computing~Discrete mathematics}
\ccsdesc[500]{Networks~Network design and planning algorithms}

\title{Realization of Temporally Connected Graphs Based on Degree Sequences\thanks{A preliminary version of this work appeared in the Proceedings of the 36th International Symposium on Algorithms and Computation (ISAAC'25), Schloss Dagstuhl – Leibniz-Zentrum für Informatik. 
This full version contains all missing proofs.}}

\keywords{
temporal graphs, temporal paths and connectivity, gossiping, (multi)graphical degree sequences, edge-disjoint spanning trees, linear time algorithms}

\begin{document}
\maketitle
\begin{abstract}
Given an undirected graph $G$, the problem of deciding whether $G$ admits a simple and proper time-labeling that makes it temporally connected is known to be NP-hard (G\"obel et al., 1991).
In this article, we relax this problem and ask whether a given degree sequence can be realized as a temporally connected graph.
Our main results are a complete characterization of the feasible cases, and a recognition algorithm that runs in $\Oh(n)$~time for graphical degree sequences (realized as simple temporal graphs) and in $\Oh(n+m)$~time for multigraphical degree sequences (realized as non-simple temporal graphs, where the number of time labels on an edge corresponds to the multiplicity of the edge in the multigraph).
In fact, these algorithms can be made constructive at essentially no cost, we present a $\Oh(n+m)$~time algorithm that outputs, for a given (multi)graphical degree sequence~$\dd$, a temporally connected graph whose underlying (multi)graph is a realization of~$\dd$, if one exists.
\end{abstract}

\maketitle

\section{Introduction}
The problem of assigning time labels to the edges of a graph in such a way that every vertex can reach every other vertex by a non-decreasing path (temporal path) is known as the gossiping problem. 
Problems of this type were extensively studied in the late 70s - early 90s (see, e.g.~\cite{hedetniemi1988survey, hromkovic1996dissemination}), initially motivated by optimal scheduling of phone calls, where the agents are represented by vertices and the calls by time-labeled edges (e.g.~\cite{hajnal_cure_1972,bumby1981problem}). 
Most of the research in this line considers interactions that are mutually exclusive (an agent cannot have multiple phone calls at the same time) and non-repetitive (two agents cannot call each other twice). 
One of the landmark results in this area was obtained in 1991 by Göbel et al.~\cite{gobel_label-connected_1991}, who showed that deciding whether a given graph admits such a labeling is \textsf{NP}-complete. 

The past decade has seen a renewed interest in this type of problems, motivated by the modeling of various types of dynamic networks, such as wireless networks, social networks, autonomous vehicles, and transportation scheduling. 
In today's terminology, the gossip problem can naturally be rephrased in terms of temporal graphs. 
Formally, a temporal graph is a pair $\G=(G,\lambda)$ where $G=(V,E)$ is the underlying graph (in this work, undirected) and $\lambda$ is a mapping that assigns time labels to the edges. 
Such a labeling is \emph{simple} if every edge of $G$ has a single time label, and it is \emph{proper} if adjacent edges do not share any time label. 
The graph $\G$ is \emph{temporally connected} (i.e. 
in class \textsf{TC}) if there exists a path traversing edges with non-decreasing times (temporal path) between each ordered pair of nodes. 
In this terminology, the gossip problem corresponds to deciding, given a graph $G$, if there exists a simple and proper labeling $\lambda$ such that $\G=(G,\lambda) \in $ \textsf{TC} (and/or to find such a labeling).

A number of similar realizability questions have been considered recently in the temporal graph literature, focusing on labelings that satisfy additional prescribed properties.
For example, several works investigate whether the labeling can ensure that the fastest temporal paths between vertices match (or are upper-bounded by) a given matrix of pairwise durations~\cite{KLOBAS2025115508, erlebach_parameterized_2025,MMMS25,MMMS25b,meusel2025directedtemporaltreerealization}. 
Another problem asks that the temporal paths realize a target reachability relation~\cite{erlebach2025recognizingrealizingtemporal}, under various restrictions of the labeling function (see~\cite{casteigts_simple_2024,doering2025simplestrictproperdirected} for discussions on the resulting expressivity).
Further studies consider minimizing the number of labels~\cite{akrida_complexity_2017}, minimizing reachability times~\cite{enright_assigning_2021}, maximizing overall reachability sets~\cite{brunelli_maximizing_2023}, or avoiding temporal cycles~\cite{andrade_temporal_2025}, to name a few.

Broadly speaking, all these problems, and their non-temporal versions before them, are part of a long series of \emph{realizability questions} in graph theory, having their root in the seminal work of Erdős and Gallai~\cite{gallai_grafok_1960} and Havel and Hakimi~\cite{havel_casopis_1955,hakimi_realizability_1962} on whether (and how, respectively) a given degree sequence can be realized as a (static) graph. 
    
In this work, we return to the roots of realizability questions, studying the temporal analog of Erdős and Gallai's problem, namely, whether (and how) a given degree sequence can be realized as a temporal graph in \textsf{TC}, for both the simple and non-simple cases. 
This question is partly motivated by the hardness results of Göbel et al.~\cite{gobel_label-connected_1991}, a natural question being how to relax the input specification in a way that offers more flexibility (and tractability) for realization. 
Another motivation is to investigate whether this format allows for an elegant and concise characterization of all the feasible cases, a situation that is pretty rare in the literature on temporal graphs.

\subsection{Main results}
Before stating the results, we recall that a \emph{degree sequence}~$\dd$ is a non-increasing list of integers. A degree sequence of length~$n$ is \emph{graphical} if there exist an $n$-vertex (static) graph whose vertices have exactly these values as their degrees. Moreover, $\dd$ is \emph{multigraphical} if a multigraph with this degree distribution exists.

Our main result is a complete characterization of the graphical sequences that admit a realization as a simple and proper temporal graph in \textsf{TC}. 
We refer to these sequences as \textit{\textsf{TC}-realizable graphical sequences}. The characterization, obtained in Section~\ref{sec:characterizations}, is the following:

\setcounter{section}{5}
\begin{restatable}[]{theorem}{main}
\label{main}
A graphical sequence~$\dd = (d_1, \dots, d_n)$ with $m:= \frac{1}{2} \cdot \sum_{i=1}^n d_i$ edges is \textsf{TC}-realizable if and only if one of the following holds:
\begin{itemize}
\item $m = 2n-4$, $d_1 < n-1$, and~$d_n \geq 2$
\item $m \geq 2n - 3$ and either (i)~$n \leq 2$  or (ii)~$n > 2$, $d_{n-1} \geq 2$, and~$d_n \geq 1$.
\end{itemize}
\end{restatable}

We also characterize the degree sequences that can be realized as a non-simple and proper temporal graph in \textsf{TC}, where \emph{non-simple} means that a same edge can receive multiple time labels. 
Conveniently, non-simple temporal graphs are equivalent to simple temporal graphs whose underlying graph is a multigraph and where each edge of the multigraph only receives one label.  
Thus, the problem amounts to characterizing \textit{\textsf{TC}-realizable multigraphical sequences}. The characterization is:

\begin{restatable}[]{theorem}{mainmulti}
\label{main 2 multi}
A multigraphical sequence~$\dd = (d_1, \dots, d_n)$ with $m:= \frac{1}{2} \cdot \sum_{i=1}^n d_i$ edges is \textsf{TC}-realizable if and only if one of the following holds:
\begin{itemize}
\item $m = 2n-4$ and $d_n \geq 2$
\item $m \ge 2n-3$ and either (i)~$n \leq 2$ 
or (ii)~$n > 2$, $d_{n-1} \geq 2$, and~$d_n \geq 1$.
\end{itemize}
\end{restatable}
\setcounter{section}{1}

As one can see, both characterizations are close to each other, the set of \textsf{TC}-realizable multigraphical sequences being naturally more general, as \textsf{TC}-realizable graphical sequences are special cases of these. However, the difference is more important than one may think by looking only at these theorems, because multigraphical sequences already differ from graphical sequences. (Both standard characterizations are recalled in the paper.)

On the algorithmic side, we show that \textsf{TC}-realizable graphical sequences can be recognized in $\Oh(n)$~time and \textsf{TC}-realizable multigraphical sequences in $\Oh(n+m)$~time. We also give a constructive $\Oh(n+m)$~time algorithm that outputs a temporal graph satisfying the desired constraints (if one exists), this algorithm being asymptotically optimal.

The fact that we consider proper labelings is not a limitation. Precisely, if non-proper labelings are considered, then one may either consider strict (i.e. increasing) or non-strict (i.e. non-decreasing) temporal paths. In the case of strict temporal paths, we know that non-proper labelings can always be turned into proper labelings without reducing the reachability relation~\cite{casteigts_simple_2024}. Together with the fact that proper labelings are particular cases of non-proper labelings, this implies that the feasible cases are the same as in the proper setting. If one considers non-strict temporal paths instead, then the problem becomes trivial to solve: a single spanning tree suffices (with the same time label on every edge). Thus, in this case, the problem reduces to testing if the sequence can be realized as a connected graph, which from earlier works~\cite{kundu_disjoint_1974,kleitman_decomposition_1976,gu_multigraphic_2012} is the case if and only if $m \geq n-1$ and $d_n\geq1$.

\subsection{Technical overview}
As already mentioned, the characterization of \textsf{TC}-realizable graphs is known to be difficult. Some necessary and sufficient conditions are known from the work of Göbel et al.~\cite{gobel_label-connected_1991} (itself based on earlier works in gossip theory~\cite{harary_communication_1974, hajnal_cure_1972,bumby1981problem}), and the remaining cases are NP-hard to decide. A necessary condition for these graphs is to contain two spanning trees that share at most two edges. 
If the graphs admit two spanning trees that share at most one edge, then this is sufficient: the graph is \textsf{TC}-realizable.
The case with two shared edges is further constrained. If the number of edges is exactly $2n-4$, then it is necessary and sufficient that the two shared edges belong to a central cycle of length~$4$.
Göbel et al.~\cite{gobel_label-connected_1991} called such graphs~``\emph{$C_4$ graphs}'' (although they do not consist only of a cycle of length 4). To avoid this confusion, we use the term \emph{\cp} graphs instead. A formal definition is given in~\Cref{sec:defs}.
Otherwise, there may exist feasible cases which are difficult to characterize, and these cases are precisely the ones making the problem \textsf{NP}-hard.
Interestingly, our results imply that the framework of degree sequences allows one to avoid these hard cases.

Our characterization proceeds as follows: First, we note that the characterization of degree sequences which are realizable as a graph admitting two edge-disjoint spanning trees is already known~\cite{kundu_disjoint_1974}. Thus, we focus on characterizing the degree sequences that can be realized as a graph admitting two spanning trees sharing exactly one edge, and, for degree sequences corresponding to exactly $2n-4$ edges, we characterize the ones admitting a realization with the above \cp graph property. 
Surprisingly, these cases already cover all the \textsf{TC}-realizable graphical sequences. Namely, if a graph has strictly more than $2n-4$ edges and is \textsf{TC}-realizable, then its degree sequence could \emph{also} be realized as a graph with two spanning trees sharing at most one edge, falling back on the previous cases. In other words, our analysis implies that the difficult graphs can safely be ignored in the framework of degree sequences. 

In light of these explanations, the feasible cases stated in~\Cref{main} correspond respectively to the following theorems in the paper: A degree sequence is \textsf{TC}-realizable if and only if it admits (i)~a realization which is a~\cp graph (\Cref{c4 sequences}) or (ii)~a realization with two spanning trees that share at most one edge (\Cref{two spanning trees}).
For multigraphical sequences, we prove the characterization similarly.

All of our proofs are constructive and rely on deconstructing recursively the degree sequence to a smaller degree sequence, using gadgets that preserve two edge-disjoint spanning trees in the constructed graph (up to a central component).
These constructive proofs can be implemented efficiently by using suitable data structures that store the degree sequence and the respective graph efficiently.
The labeling is then handled by a dedicated procedure that achieves the claimed time using our data structures and extra features offered by the above algorithms.

Temporal graphs are notoriously intractable objects. Most of the results in this young field are negative and most of the problems turn out to be computationally hard, often due to the fact that the temporal reachability relation is neither symmetric nor transitive. In this respect, the fact the feasible cases for \textsf{TC}-realizability admit a characterization that is at the same time compact, purely structural, and easy to recognize is quite significant. This situation is clearly an exception in the landscape of temporal graph theory.

\section{Definitions and Important Existing Results}\label{sec:defs}

A temporal graph is a pair $\G=(G,\lambda)$, where $G=(V,E)$ is a standard (in this work, \emph{undirected}) graph called the \emph{underlying graph} of $\G$, and $\lambda:E \to 2^{\mathbb{N}}$ is a labeling function that assigns a non-empty set of time labels to every edge of $E$, interpreted as availability times. 
The labeling function $\lambda$ can be restricted in various ways. It is called \emph{proper} if adjacent edges cannot share a common time label ($\lambda$ is locally injective), and it is called \emph{simple} if every edge has exactly one time label ($\lambda$ is single-valued). The typical setting of gossiping, including that of Göbel et al.~\cite{gobel_label-connected_1991}, requires that the labeling is both proper and simple.

A pair $(e,t)$ such that $t \in \lambda(e)$ is a \emph{temporal edge} of \G. 
A \emph{temporal path} in $\G$ is a sequence of temporal edges $\langle (e_i,t_i)\rangle$ such that $\langle e_i \rangle$ is a path in the underlying graph and $\langle t_i \rangle$ is non-decreasing (such a path is \emph{strict} if $\langle t_i \rangle$ is increasing).
A temporal graph \G is \emph{temporally connected} (in class \textsf{TC}) if temporal paths exist between all ordered pairs of nodes.
Observe that, beyond modeling mutually exclusive interaction, the proper setting has the technical advantage of removing the distinction between strict and non-strict temporal paths (indeed, all the temporal paths in this case are \emph{de facto} strict), which allows us to rely on a single definition of \textsf{TC} throughout the paper. 

We can now state the gossiping problem as follows:

\prob{\textsc{Gossiping}}
{An undirected graph $G=(V,E)$.}
{Does there exist a simple and proper labeling $\lambda$ such that $(G,\lambda)\in \textsf{TC}$?}

 A \emph{degree sequence} is a sequence $\dd = (d_1, \dots, d_n)$ of non-negative integers. By convention, we always require that this sequence is ordered non-increasingly.  
The original realizability question of Erdős and Gallai asks the following:

\prob{\textsc{Degree Sequence Realization}}
{A degree sequence~$\dd = (d_1, \dots, d_n)$.}
{Does there exist a graph $G=(V,E)$ with $V=\{v_1,\dots,v_n\}$ such that $v_i$ has degree $d_i$?}

When the answer is yes, the degree sequence $\dd$ is called \emph{graphical} and $G$ is called a \emph{realization} of $\dd$.
If multiedges are permitted in $G$ and if there is a multigraph that realizes~$\dd$, then the sequence is instead called \emph{multigraphical}~\cite{mcavaney_simple_1981}.
In this work, we study the following temporal version of the problem:

\prob{\textsc{Degree Sequence TC Realization}}
{A degree sequence~$\dd = (d_1, \dots, d_n)$.}
{Does there exist a realization $G=(V,E)$ of $\dd$ and a labeling $\lambda$ such that $(G,\lambda) \in \textsf{TC}$?}

Next, we state important existing results on (multi)graphical degree sequences and TC-realizable graphs and derive immediate consequences that we use throughout the paper.
\subsection{Graphical Degree Sequences}

To be \emph{graphical}, a degree sequence must satisfy two properties: (i) the sum of all degrees must be even (standard handshaking lemma), and (ii) to avoid the need for multi-edges, the sum of degrees of $r$ vertices cannot exceed what these vertices can actually absorb; namely, at most $r(r-1)$ for the edges among themselves (sum of degrees in a complete graph on $r$ vertices), plus at most $\min(r, d_i)$ for the edges to each remaining vertex~$i$.
This is formalized in the following classical characterization.

    \begin{theorem}[Graphical sequence \cite{gallai_grafok_1960,tripathi_short_2010}] \label{def: simple graphical}
    A degree sequence~$\dd =  (d_1, \dots, d_n)$ is graphical if and only if \\
    \begin{enumerate*} [label=(\roman*)]
        \item $\displaystyle \sum_{i=1}^n d_i$ is even, and
        \item for each~$r\in [1,n-1]:$\quad
        $\displaystyle\sum_{i=1}^r d_i \leq r \cdot (r-1) + \sum_{i=r+1}^n \min(r, d_i).$
    \end{enumerate*} 
                        \end{theorem}
    
    Based on this characterization, we obtain the following sufficient condition under which degree sequences with few edges are graphical.
    
    \begin{corollary}\label{realizable if dn geq 3}
        
    Let~$\dd =  (d_1, \dots, d_n)$ be a degree sequence with~$\sum_{i=1}^n d_i$ being even, $\sum_{i=1}^n d_i \leq 4(n-1)$, $n > 4$, $d_1 \leq n-1$, and~$d_n \geq 2$.
    Then~$\dd$ is graphical if $d_4\geq 3$.
    \end{corollary}
    
    \begin{proof}
    According to~\Cref{def: simple graphical}, it suffices to show that for each~$r\in [1,n-1]$, $\sum_{i=1}^r d_i \leq r \cdot (r-1) + \sum_{i=r+1}^n \min(r, d_i)$.
    For~$r = 1$, the statement follows since~$d_1 \leq n-1 = (n-1) \cdot 1$.
    
    Consider~$r = 2$ and observe that~$d_1+d_2 \leq 2(n-1)$. 
    Moreover, $\sum_{i=3}^n \min(r=2, d_i) = 2(n-2)$ since~$d_n \geq 2$.
    Hence, $d_1 + d_2 \leq 2(n-1) = 2 + 2(n-2) = 2  + \sum_{i=3}^n \min(r=2, d_i)$, showing the statement for~$r=2$.
    
    Now, consider~$r = 3$.
    Recall that~$\sum \dd = 4(n-1) - c \leq 4(n-1)$, $d_n \geq 2$, and~$d_4 \geq 3$.
    This implies that~$\sum_{i=4}^n d_i \geq \sum_{i=4}^n \min(r=3,d_i) \geq 3 + (n-4)\cdot 2 = 2(n-1) - 3$.
    Thus, $d_1+d_2+d_3 = \sum \dd - \sum_{i=4}^n d_i \leq 4(n-1) - (2(n-1) - 3) = 2(n-1) + 3$. 
    Consequently, $d_1+d_2+d_3 \leq 2(n-1) + 3 = 6 + 2(n-1) - 3 \leq r\cdot (r-1) + \sum_{i=4}^n \min(r,d_i)$, showing the statement for~$r=3$.  
        
    Finally, consider~$r \geq 4$.
    Recall that~$\sum\dd \leq 4(n-1)$ and~$d_n \geq 2$. The latter implies $\sum_{i=r+1}^n \min(r,d_i) \geq 2 (n-r)$. Consequently, $\sum_{i=1}^r d_i = \sum \dd - \sum_{i=r+1}^n d_i\leq 4(n-1) - 2(n-r) = 2(n-1) + 2(r-1)$.
    Moreover, $r\cdot (r-1)\geq 4(r-1)$ since~$r\geq 4$.
    This implies that~$\sum_{i=1}^r d_i \leq 2(n-1) + 2(r-1) = 4(r-1) + 2(n-r) \leq r\cdot (r-1) + \sum_{i=r+1}^n \min(r,d_i) \geq 2 (n-r)$, showing the statement for~$r\geq 4$.    
    \end{proof}
    
    We will make use of this corollary several times in this work to easily argue that specific degree sequences are graphical.
    Similarly, we also obtain the following sufficient condition.
    \begin{corollary}\label{realizable degree 1 and bound}
        Let~$\dd =  (d_1, \dots, d_n)$ be a degree sequence with~$\sum_{i=1}^n d_i = 4(n-1) - 2$, $n > 4$, $d_1 \leq n-1$, $d_{n-1} \geq 2$, and~$d_n \geq 1$.
        Then~$\dd$ is graphical if $d_4 \geq 3$.
    \end{corollary}
    
    \begin{proof}[Sketch]
        The statement can be shown analogously to~\Cref{realizable if dn geq 3}.
        The only differences are that for each~$r\in [2,n-1]$, $\sum_{i=r+1}^n \min(r,d_i)$ is smaller by at most one because~$d_n = 1$ and~$d_{n-1} \geq 2$, and that $\sum_{i=1}^r d_i$ is smaller by at least one since~$\sum\dd$ is reduced by 2 compared to~\Cref{realizable if dn geq 3}.
        Hence, $\sum_{i=1}^r d_i \leq r\cdot (r-1)  + \sum_{i=r+1}^n \min(r,d_i)$ holds for each~$r\in [1,n-1]$.
    \end{proof}
    
    The following algorithm constructs a graph from a graphical degree sequence.
                                                \begin{definition} [Graphical Laying Off Process \cite{hakimi_realizability_1962}] \label{def:havel hakimi laying off simple}
        Let $\dd = (d_1, \dots, d_n)$ be a graphical sequence.
        The \emph{laying off} procedure consists of connecting the vertex $v_i$ to the first $d_i$ vertices, excluding $v_i$. The resulting \emph{residual sequence} is given by:
        \begin{align*}
            &(d_1-1,\dots,d_{d_i}-1,d_{d_i+1},\dots,d_{i-1},d_{i+1},\dots,d_n) &\text{if } d_i<i, \\
            &(d_1-1,\dots,d_{i-1}-1,d_{i+1}-1,\dots,d_{d_i+1}-1,d_{d_i+2},\dots,d_n) 
            &\text{if } d_i\geq i.
        \end{align*}
    \end{definition}
    The next lemma implies that the laying off process can be used iteratively to build a realization for a graphical degree sequence.
    \begin{lemma}[\cite{fulkerson_properties_1965,hakimi_realizability_1962}] \label{lem:havel-hakimi-correctness}
        If $\dd=(d_1,\dots,d_n)$ is a graphical sequence, then the residual sequence after laying off any entry is also graphical.
    \end{lemma}

\subsection{Multigraphical degree sequences}
To be \emph{multigraphical}, a degree sequence should also satisfy that the sum of degrees are even (handshaking lemma again), but since multi-edges are permitted, the other terms are not as restricted as for graphical sequences. The only restriction is that the highest-degree vertex can distribute all its edges among the remaining vertices.
This is formalized as follows.

    \begin{theorem}
    [Multigraphical sequence \cite{hakimi_realizability_1962}]\label{multi realize chara}
    A degree sequence~$\dd =  (d_1, \dots, d_n)$ is multigraphical if and only if  \begin{enumerate*} [label=(\roman*)]
            \item $\sum_{i = 1}^n d_i$ is even, and 
            \item $\displaystyle d_1 \leq \sum_{i = 2}^n d_i$.  
        \end{enumerate*} 
                                            \end{theorem}
    For multigraphical sequences there also exists a ``laying off'' process which is similar to \Cref{def:havel hakimi laying off simple}, but more flexible. 
    Contrary to the graphical case, where a whole degree is laid off, in the multigraphical case, only a single edge is laid off.
    \begin{theorem}[Multigraphical Laying Off Process \cite{boesch_line_1976}]\label{laying off multi}
Let $\dd=(d_1,\dots,d_n)$ be a multigraphical degree sequence with $d_n>0$.
Then, for each~$j, 2\leq j\leq n$, the degree sequence $(d_1-1,d_2,\dots,d_{j-1},d_j-1,d_{j+1},\dots,d_n)$ is multigraphical (after reordering).
    \end{theorem}

\subsection{(Multi)graphical sequences admitting two edge-disjoint spanning trees}

    The characterization of (multi)graphical sequences that admit a realization containing $k$ edge-disjoint spanning trees is identical for graphical and multigraphical sequences.
    For graphical sequences, this result was first established for the special case of $k=2$ by Kundu~\cite{kundu_disjoint_1974} through a constructive proof, and two years later generalized to arbitrary $k$ by Kleitman and Wang~\cite{kleitman_decomposition_1976}, also constructively. For multigraphical sequences, this statement was proven non-constructively by Gu, Hai and Liang~\cite{gu_multigraphic_2012}. We are primarily interested in the case $k=2$. Nevertheless, we state the full characterization for general $k$ below.
    
    \begin{theorem}[\cite{kundu_disjoint_1974,kleitman_decomposition_1976,gu_multigraphic_2012}]
    \label{thm:edt_76}\label{two spanning trees multi sequence}\label{thm:two-edst-kundu}
        A (multi)graphical sequence~$\dd =  (d_1, \dots, d_n)$ with~$n\geq 2$ admits a (multigraphical) realization with~$k\in\mathbb{N}$ edge-disjoint spanning trees if and only if\\ \begin{enumerate*} [label=(\roman*)]
            \item $d_n \geq k$, and
            \item $\displaystyle\sum_{i=1}^n d_i\geq 2\cdot k(n-1)$.
        \end{enumerate*} 
    \end{theorem}
                                                                                        
\subsection{TC-realizable graphs}
Finally, we provide an overview on some known necessary and sufficient conditions of TC-realizable (multi)graphs.
\begin{lemma}[\cite{hajnal_cure_1972}]\label{min number edges of tc graphs}
Let~$G$ be a TC-realizable (multi)graph on~$n$ vertices.
Then~$G$ has at least~$2n-4$ edges.
\end{lemma}
\begin{lemma}[\cite{baker_gossips_1972}]\label{one edge towards EDSTs minus one}
Let~$G$ be a (multi)graph with two spanning trees that share at most one edge.
Then~$G$ is TC-realizable.
\end{lemma}

\begin{definition}[Reformulation of Göbel et al.~\cite{gobel_label-connected_1991}]
A graph~$G$ on~$n$ vertices and~$2n-4$ edges is a~\emph{\cp graph} if~$G$ contains an induced cycle~$C$ of length~$4$ (called a~\emph{central cycle}) and two spanning trees that share exactly two edges, where both shared edges are from~$C$.
\end{definition}

This definition is equivalent to the one provided by Göbel et al.~\cite{gobel_label-connected_1991}, since no tree can contain more than three edges of the central cycle, and due to the number of edges, each edge of the cycle is contained in at least one of the two spanning trees. 
\begin{lemma}[\cite{gobel_label-connected_1991}]\label{c4 property}
Let~$G$ be a graph with~$n$ vertices and~$2n-4$ edges.
Then~$G$ is TC-realizable if and only if~$G$ is a~\cp graph.
\end{lemma}

Even though it is not explicitly stated~\cite{gobel_label-connected_1991}, the following generalization of~\cp graphs to multigraphs is also TC-realizable by the same labeling procedure used for~\cp graphs.\footnote{We will recall this labeling procedure in~\Cref{sec:algos} and provide an efficient algorithm.}

\begin{definition}
A multigraph~$G$ on~$n$ vertices and~$2n-4$ edges is a~\emph{\cp multigraph} if~$G$ contains an induced cycle~$C$ of length~$4$ (called a~\emph{central cycle}) and two spanning trees that share exactly two edges, where both shared edges are from~$C$.
\end{definition}
Here, an induced cycle in a multigraph~$G$ has the additional property that each edge of~$C$ exists exactly once in~$G$. 
Note that a~\cp graph is also a~\cp multigraph.

\begin{corollary}\label{c4 property multi}
Let~$G$ be a~\cp multigraph.
Then~$G$ is TC-realizable.
\end{corollary}

\section{Realizations with Two Spanning Trees Sharing at Most one Edge}\label{sec:twotrees}
In this section, we establish the following characterization of sequences that allow for a realization with two spanning trees that share (at most) one edge.
\begin{theorem}\label{two spanning trees}
Let~$\dd = (d_1, \dots, d_n)$ be a graphical sequence.
Then, $\dd$ admits a realization with two spanning trees that share at most one edge if and only if
$\sum_{i=1}^n d_i \geq 4(n-1)-2$ and (a)~$n\leq 2$ or (b)~$n > 2$, $d_{n-1} \geq 2$, and~$d_n \geq 1$.
\end{theorem}

First, we consider the special cases of~$n\leq 2$.
To this end, note that~$\{(0),(0,0),(1,1)\}$ are exactly the graphical sequences of length at most~$2$.
Each of these sequences has a unique realization and~$(0,0)$ is the only sequence for which the unique realization is not connected.
Hence for $\dd = (0,0)$, 
    there is no realization with two spanning trees that share at most one edge and 
    $\sum\dd = 0 < 2 = 4(n-1)-2$.
For~$\dd\in \{(0),(1,1)\}$,  
    $\sum\dd \geq 4(n-1)-2$ and 
    there is no realization with two spanning trees that share at most one edge.
This proves the equivalence for part (a) of \Cref{two spanning trees}.

\begin{figure}[t]
    \centering
    \includegraphics[width=0.9\linewidth]{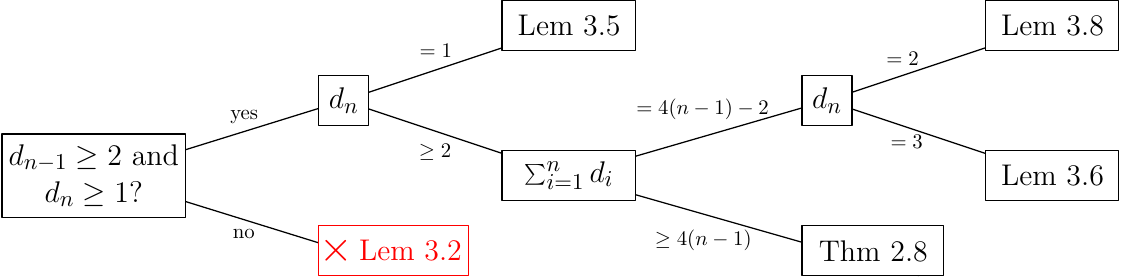}
    \caption{A guideline for the branching cases in the proof of \Cref{two spanning trees}.}
    \label{fig:Overview section EDST-1}
\end{figure}
By the above, we may assume $n>2$ and proceed to show the forward direction for part (b) of \Cref{two spanning trees}.
\Cref{fig:Overview section EDST-1} provides a guideline for the steps of the proof of this case.

\begin{lemma}\label{no two trees if small degree}

Let~$\dd = (d_1, \dots, d_n)$ be a graphical sequence with~$n>2$.
If~$\dd$ admits a realization with two spanning trees that share at most one edge, then~$\sum_{i=1}^n d_i \geq 4(n-1)-2$, $d_{n-1} \geq 2$, and~$d_n \geq 1$.
\end{lemma}

\begin{proof}
Note that each graph on~$n > 2$ vertices needs at least~$2(n-1) - 1$ edges to have two spanning trees that share at most one edge.
This immediately implies that~$\sum \dd \geq 4(n-1)-2$.

It remains to show that~$d_{n-1} < 2$ or~$d_n = 0$.
We show this statement via contraposition.

If~$d_n = 0$, then each realization~$G$ of~$\dd$ contains an isolated vertex.
Since~$n>2$, $G$ is not connected and does not even have a single spanning tree.
This implies that~$d_n \geq 1$.

If~$d_n = d_{n-1} = 1$, then each realization~$G$ of~$\dd$ contains two degree-1 vertices~$v_1$ and~$v_2$.
Let~$e_1$ denote the single edge of~$G$ incident with~$v_1$ and let~$e_2$ denote the single edge of~$G$ incident with~$v_2$.
In the case that~$v_1v_2$ is an edge of~$G$, this edge forms a connected component, which implies that $G$ is not connected since $n>2$. 
Hence, $G$ admits no spanning tree.
In the case that~$v_1v_2$ is not an edge of~$G$, every spanning tree of~$G$ must contain both edges~$e_1$ and~$e_2$, since~$v_1$ and~$v_2$ have degree 1.
By the fact that~$e_1 \neq e_2$, this implies that~$G$ does not admit two spanning trees that share at most one edge.
Consequently, $d_{n-1}\geq 2$.
\end{proof}

\subsection{Realizability} Now, we consider the backward direction for part (b) of~\Cref{two spanning trees}.
That is, we will show the following.
\begin{lemma}
Let~$\dd = (d_1, \dots, d_n)$ be a graphical sequence with~$\sum_{i=1}^n d_i \geq 4(n-1)-2$, $n>2$, $d_{n-1} \geq 2$, and~$d_n \geq 1$.
Then~$\dd$ admits a realization with two spanning trees that share at most one edge.
\end{lemma}
First, recall~\Cref{thm:two-edst-kundu}, for which the following is a corollary.
\begin{corollary}[\cite{kleitman_decomposition_1976}]
Let~$\dd = (d_1, \dots, d_n)$ be a graphical sequence with $\sum_{i=1}^n d_i \geq 4(n-1)$, $n > 2$, and $d_n \geq 2$. 
Then~$\dd$ admits a realization that has two edge-disjoint spanning trees.
\end{corollary}

Hence, it remains to consider (i)~$\sum\dd \geq 4(n-1) - 2$, $d_{n-1}\geq 2$, and~$d_n=1$, and (ii)~$\sum\dd = 4(n-1) - 2$ and~$d_n \geq 2$.
We first consider the former. 

\begin{lemma}\label{tc if deg 1}
Let~$\dd = (d_1, \dots, d_n)$ be a graphical sequence with $\sum_{i=1}^n d_i \geq 4(n-1) - 2$, $n > 2$, $d_{n-1} \geq 2$, and~$d_n = 1$. 
Then~$\dd$ admits a realization with two spanning trees that share one edge.
\end{lemma}

\begin{proof}
Since~$\dd$ is graphical, $\sum\dd$ is even.
Hence, $d_1 > 2$, since~$d_n = 1$ and~$d_{n-1}\geq 2$.
This further implies that~$n \geq 4$, as otherwise~$d_1 \geq 3 = n$.

Consider the residual sequence~$\dd'$ after laying off entry~$d_n$ in~$\dd$.
Recall that this sequence is obtained by removing~$d_n$ from~$\dd$ and decreasing the value of the first~$d_n = 1$ entries of~$\dd$ by 1 each (and reordering the entries).
Due to~\Cref{lem:havel-hakimi-correctness}, $\dd'$ is graphical.
Hence, $\dd' = (d_1', \dots, d_{n-1}')$ is a graphical sequence of length~$n' := n-1$ that fulfills~$\sum\dd' \geq 4(n-1) - 2 - 2 = 4(n-2) = 4(n'-1)$ and~$d'_{n'}=\min(d_1-1,d_{n-1}) \geq 2$, since~$d_1 \geq 3$.
This implies that there is a realization~$G'$ for~$\dd'$ that contains two edge-disjoint spanning trees~$T'_1$ and~$T'_2$ according to~\Cref{thm:two-edst-kundu}.

Now consider the graph~$G$ obtained from~$G'$ by adding a degree-1 neighbor~$v_n$ to an arbitrary vertex~$v$ of~$G'$ of degree~$d_1-1$.
Then, $G$ is a realization of~$\dd$.
Moreover, $G$ contains the spanning trees~$T_1 := T'_1 \cup vv_n$ and~$T_2 := T'_2 \cup vv_n$ that only share the edge~$vv_n$, since~$T'_1$ and~$T'_2$ are edge-disjoint.
\end{proof}

Now, we consider~$\sum\dd = 4(n-1) - 2$ and~$d_n \geq 2$.
Note that these sequences are guaranteed to fulfill~$d_n \in \{2,3\}$, as otherwise, $\sum \dd \geq 4 n > 4(n-1) - 2$.
We first prove the case of~$d_n=3$ in \Cref{minus two dn eq 3}, which is needed for the subsequent case~$d_n=2$ in \Cref{minus two dn eq 2}.
 
\begin{lemma}\label{minus two dn eq 3}
Let~$\dd = (d_1, \dots, d_n)$ be a graphical sequence with $\sum_{i=1}^n d_i = 4(n-1) - 2$ and~$d_n = 3$. 
Then~$\dd$ admits a realization with two spanning trees that share one edge.
\end{lemma}
\begin{proof}
We will show that~$n\geq 6$ in this case.
To this end, we provide a lower bound on the number of entries of value~$3$ in~$\dd$.

\begin{claim}\label{claim number of threes}

Let~$\ell$ be the smallest index of~$[1,n]$ for which~$d_\ell = 3$. 
Then, the number of entries of value~$3$ in~$\dd$ is~$6 + \sum_{i=1}^{\ell-1} (d_i - 4)$.
\end{claim}

\begin{claimproof}
In other words we want to show that~$n$ equals~$\ell - 1 + 6 + \sum_{i=1}^{\ell-1} (d_i - 4)$.
Since~$\sum\dd = 4(n-1)-2$, $\sum\dd = \sum_{i=1}^{\ell-1} d_i  + (n-(\ell-1))\cdot 3 = 4(n-1)-2 = 4(n-6) + 3\cdot 6$.
This implies that~$\ell \leq n-5$ and moreover, that~$\sum_{i=1}^{\ell-1} d_i + 3\cdot (n-(\ell-1) - 6) = 4(n-6)$.
Thus, $\sum_{i=1}^{\ell-1} d_i = 3(\ell-1) + n-6 = 4(\ell-1) + (n-(\ell-1)-6)$, which leads to~$\sum_{i=1}^{\ell-1} (d_i - 4) = (n-(\ell-1)-6)$.
Hence, the statement holds.
\end{claimproof}

Note that this implies that~$d_1 < n-1$, as otherwise, $\dd$ would contain at least $n-1 - 4 + 6 > n$ entries of value~$3$.
We distinguish between~$d_2 = 3$ and~$d_2 \geq 4$.

Firstly, consider the case that~$d_2 = 3$.
Note that this implies by~\Cref{claim number of threes} that~$d_1 = n-3$, since~$\sum\dd = 4(n-1)-2$.
In other words, for each~$n\geq 6$, there is exactly one degree sequence~$\dd^n := (d^n_1, \dots, d^n_n)$ of length~$n$ with~$d^n_2 = d^n_n = 3$ fulfilling~$\sum\dd^n = 4(n-1)-2$.
Realizations with two spanning trees that share one edge for the three sequences~$\dd^6$, $\dd^7$, and~$\dd^8$ are depicted in \Cref{fig:1 d6 d7 d8}.

\begin{figure}
    \centering
    \includegraphics[width=0.85\linewidth]{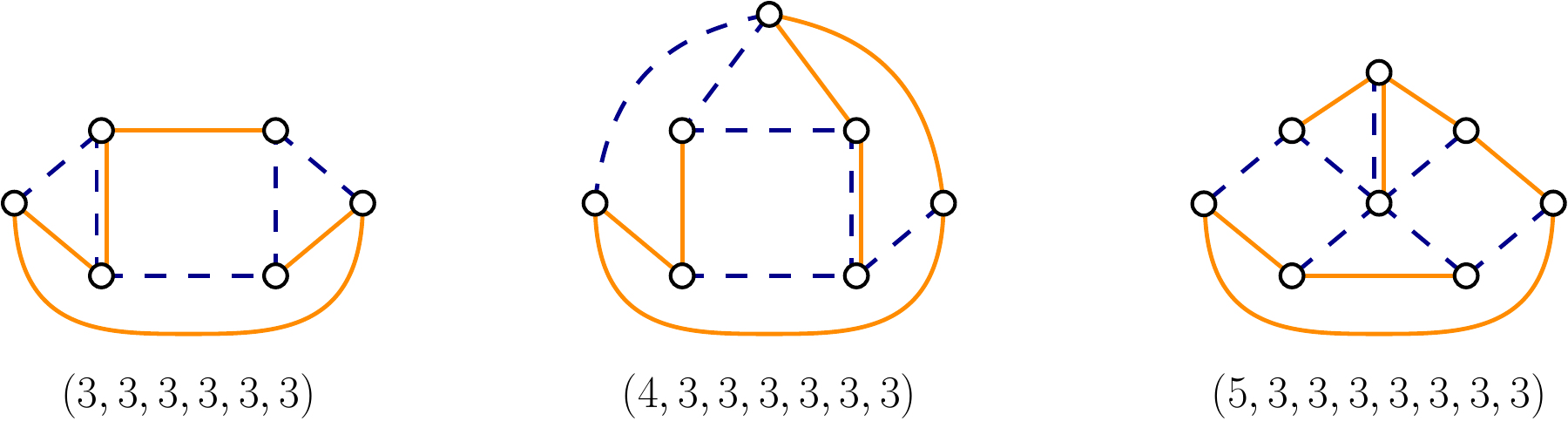}
    \caption{Realizations of the degree sequences $\dd^6=(3,3,3,3,3,3)$, $\dd^7=(4,3,3,3,3,3,3)$, and~$\dd^8=(5,3,3,3,3,3,3,3)$ with two spanning trees $T_1$ (solid orange edges) and $T_2$ (dashed blue edges) that share exactly one edge, shown left to right.}
    \label{fig:1 d6 d7 d8}
\end{figure}

Now, consider~$n\geq 9$.
Recall that~$d_1 = n-3$ and $d_i = 3$ for each~$i\in [2,n]$.
Moreover, recall that~$\dd^6 = (3,3,3,3,3,3)$ has a realization~$G'$ with two spanning trees~$T_1'$ and~$T_2'$ that share exactly one edge (see~\Cref{fig:1 d6 d7 d8}).
Let~$v_1$ be an arbitrary vertex of degree~$3$ in~$G'$.
We obtain a graph~$G$ by adding a cycle~$C$ of length~$d_1-3 = n-6$ to~$G'$ and adding the edge~$v_1x$ to~$G'$ for each vertex~$x$ of the cycle~$C$.
Note that~$G$ is a realization of~$\dd$, since~$G$ contains~$5 + d_1-3 = n-1$ vertices of degree~$3$ and the degree of~$v_1$ was increased by~$d_1-3$ to~$d_1$.
Since~$G^c := G[\{v_1\}\cup C]$ consists of a universal vertex attached to a cycle of length at least 3 (since~$d_1 - 3 = n-6 \geq 3$ by~$n\geq 9$), $G^c$ contains two edge-disjoint spanning trees~$T^c_1$ and~$T^c_2$.
Hence, $T_1 := T'_1 \cup T^c_1$ and~$T_2:= T'_2 \cup T^c_2$ are spanning trees of~$G$, and~$T_1$ and~$T_2$ share only one edge, namely the edge shared by~$T_1'$ and~$T_2'$.
Hence, $\dd$ admits a realization with two spanning trees that share one edge.

Secondly, consider the case that~$d_2 \geq 4$. 
Note that~\Cref{claim number of threes} thus implies that~$n \geq d_1 + 4$, since~$d_2\geq 4$ and~$\dd$ contains at least~$d_1 - 4 + 6 = d_1 + 2$ entries of value~$3$.
Let~$\dd'$ be the sequence obtained from~$\dd$ by removing the entry~$d_1$, removing~$d_1-1$ entries of value 3, and adding a 1.
We first show that~$\dd'$ admits a realization with two spanning trees that share one edge, and then describe how to construct the realization for~$\dd$.

Note that~$\dd'$ has length~$n' := n-(d_1-1)$ and~$\sum\dd' = \sum\dd - d_1 - (d_1-1)\cdot 3 + 1 = 4(n-1) - 2 - 4(d_1-1) = 4(n - d_1) - 2 = 4(n'-1)-2$.
Moreover, $d'_4 \geq 3$ since~$n' = n-(d_1-1) \geq d_1 + 4 - (d_1-1) = 5$.
Due to~\Cref{realizable degree 1 and bound}, $\dd'$ is graphical if~$d'_1 \leq n'-1$.
To see that this is the case, suppose that~$d'_1 > n'-1$.
By definition of~$\dd'$, $d'_1 = d_2$.
This would then imply that~$\sum\dd \geq d_1 + d_2 + (n-2)\cdot 3 \geq d_1 + n' + 3(n-2) = d_1 + n-(d_1-1) + 3(n-2) = n + 1 + 3(n-2) = 4(n-1) - 1 > 4(n-1) - 2$, which is not the case by assumption.
Hence, $d'_1 \leq n'-1$, which implies that~$\dd'$ is graphical due to~\Cref{realizable degree 1 and bound}. 
Moreover, since~$\dd'$ fulfills~$n' > 2$, $d_{n'-1}' \geq 2$, $d_{n'}' = 1$, and~$\sum\dd' = 4(n'-1)-2$, \Cref{tc if deg 1} implies that~$\dd'$ admits a realization~$G'$ with two spanning trees~$T_1'$ and~$T_2'$ that share one edge.

Now consider the graph~$G$ obtained by adding  a cycle~$C$ of length~$d_1-1$ to~$G'$ and adding an edge between~$v_1$ and each vertex of~$C$, where~$v_1$ is the unique vertex of~$G'$ of degree 1.
Then $G$ is a realization for~$\dd$, since  $v_1$ has degree $1 + (d_1-1) = d_1$ in~$G$ and~$d_1-1$ additional vertices of degree 3 were added to~$G'$ (namely, the vertices of~$C$).
It suffices to show that~$G$ has two spanning trees that share one edge.
Since~$G^c := G[\{v_1\}\cup C]$ consists of a universal vertex attached to a cycle of length at least 3 (since~$d_1 \geq 4$), $G^c$ contains two edge-disjoint spanning trees~$T^c_1$ and~$T^c_2$.
Hence, $T_1 := T'_1 \cup T^c_1$ and~$T_2:= T'_2 \cup T^c_2$ are spanning trees of~$G$, and~$T_1$ and~$T_2$ share only one edge, namely the edge shared by~$T_1'$ and~$T_2'$.
Hence, $\dd$ admits a realization with two spanning trees that share one edge. 
\end{proof}

Finally, we consider the case where~$\sum\dd=4(n-1)-2$ and $d_n=2$.
Essentially, we perform an induction over the length of the sequence and lay off the smallest degree~$d_n = 2$, until we reach a sequence of constant length or~$d_n$ has value unequal to~$2$.
In the latter case, the previous lemma acts as a base case.

\begin{lemma}\label{minus two dn eq 2}

Let~$\dd = (d_1, \dots, d_n)$ be a graphical sequence with $\sum_{i=1}^n d_i = 4(n-1) - 2$ and~$d_n = 2$. 
Then~$\dd$ admits a realization with two spanning trees that share one edge.
\end{lemma}

\begin{proof}
To prove the statement, we show that for each~$n$, each graphical sequence~$\dd^n$ of length~$n$ with $\sum \dd^n = 4(n-1) - 2$ and~$d_n = 2$ admits a realization with two spanning trees that share one edge.
We show this statement via induction over~$n$.

For the base case, note that for~$n \in \{1,2\}$, there is no graphical sequence with~$d_n = 2$. 
Moreover, for~$n = 3$, $(2,2,2)$ is the only graphical sequence of length~$n$ with~$d_n = 2$.
This sequence is TC-realizable, since the unique realization of~$\dd$ is a triangle, which obviously contains two spanning trees that share exactly one edge.

For the inductive step, let~$n \geq 4$ and assume that the statement holds for~$n-1$.
Let~$\dd$ be an arbitrary graphical sequence of length~$n$ fulfilling~$\sum\dd = 4(n-1)-2$ and~$d_n = 2$.
We show that~$\dd$ admits a realization with two spanning trees that share one edge.
Since~$\dd$ is graphical, the residual sequence~$\dd' = (d_1', \dots, d_{n'}')$ obtained from laying off~$d_n$ is graphical (see~\Cref{lem:havel-hakimi-correctness}).
Moreover, $\dd'$ has length~$n' = n-1$ and fulfills~$\sum \dd' = \sum \dd - 4 = 4(n-1)-2 - 4 = 4(n'-1)-2$, since $d_n=2$.
Hence, $\dd'$ admits a realization~$G'$ with two spanning trees~$T_1'$ and~$T_2'$ that share one edge by (i)~the induction hypothesis if~$d'_{n'} = 2$ and (ii) by~\Cref{minus two dn eq 3} if~$d'_{n'} = 3$.
Note that~$d'_{n'} < 2$ is not possible, since~$d_2 > 2$.
The latter follows from the fact that for~$n\geq 4$, $\sum \dd = 4(n-1)-2$ and~$d_{n} = 2$.
Let~$v_1$ be a vertex of degree~$d_1-1$ in~$G'$, and let~$v_2$ be a vertex of degree~$d_2-1$ in~$G'$.
These vertices exists, since~$G'$ is a realization of~$\dd'$ and~$\dd'$ is obtained from~$\dd$ by laying off~$d_n=2$.
Consider the graph~$G$ obtained by adding a vertex~$v_n$ to~$G'$ as well as the edges~$v_1v_n$ and~$v_2v_n$.
Note that~$G$ is a realization of~$\dd$.
Moreover, $T_1 := T'_1 \cup v_1v_n$ and~$T_2 := T'_2 \cup v_2v_n$ are spanning trees of~$G$ that share exactly one edge, namely the unique edge shared by~$T_1'$ and~$T_2'$.
This implies that~$\dd$ admits a realization with two spanning trees that share one edge.
Since~$\dd$ is an arbitrary graphical sequence of length~$n$ that fulfills~$\sum\dd = 4(n-1)-2$ and~$d_n = 2$, the statement holds for~$n$.
\end{proof}

\section{Realizations of \cp graphs}\label{sec:c4}
\begin{figure}[t]
    \centering
    \includegraphics[width=0.55\linewidth]{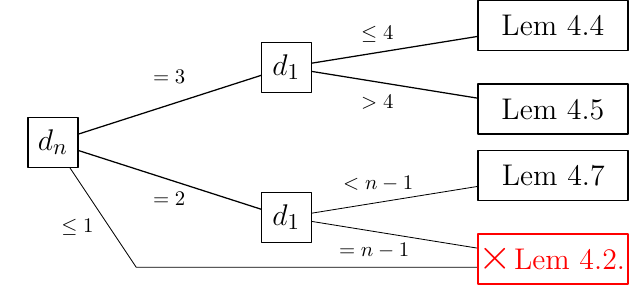}
    \caption{A guideline for the branching cases in the proof of \Cref{c4 sequences}.}
    \label{fig:Overview section C4}
\end{figure}
In this section, we establish the following characterization of graphical sequences that allow for a~\cp graph realization.

\begin{theorem}\label{c4 sequences}
Let~$\dd = (d_1, \dots, d_n)$ be a graphical sequence.
Then, $\dd$ admits a realization which is a~\cp graph if and only if
$\sum_{i=1}^n d_i = 4(n-1)-4$, $d_1 < n-1$, and~$d_n \geq 2$.
\end{theorem}
\Cref{fig:Overview section C4} provides a guideline on the individual steps for proving~\Cref{c4 sequences}.
First, we show the forward direction of~\Cref{c4 sequences}.
To this end, we analyze the minimum and maximum degree of~\cp graphs.

\begin{lemma}\label{c4 degree bounds}

Let~$G$ be a~\cp graph on~$n$ vertices. 
Then~$G$ has a minimum degree of at least~$2$ and a maximum degree of at most~$n-2$.
\end{lemma}

\begin{proof}
Recall that by definition of a~\cp graph, $G$ has~$2n-4$ edges and at least~$4$ vertices.

By definition, each vertex~$v$ of the~\cp graphs~$G$ is (i)~part of the central cycle or (ii)~a vertex that has at least one incident edge in each of the two spanning trees.
In both cases, vertex~$v$ has degree at least 2.
Hence, the minimum degree of~$G$ is at least 2.

We now show that~$G$ also has a maximum degree of at most~$n-2$.
We show this statement via induction over~$n$.

For~$n = 4$, there is only a single~\cp graph, namely the cycle of length 4 itself, which has a maximum degree of~$2 = n-2$.
For~$n = 5$, the only (non-isomorphic) graph with minimum degree at least~$2$ and maximum degree~$4 = n-1$ that has~$2n-4 = 6$ edges is the graph~$G'$ that contains a universal vertex~$v_1$ and where~$G' - v_1$ is a perfect matching of size~$2$.
Hence, $G'$ contains no induced cycle of length~$4$, which implies that~$G'$ is not a~\cp graph.

For the inductive step, assume that~$n>5$ and that for $n'=n-1$, there is no~\cp graph on~$n'$ vertices with maximum degree~$n'-1$.

Assume towards a contradiction that there is a~\cp graph~$G_n$ on~$n$ vertices that has a maximum degree of~$n-1$.
Then, the minimum degree of~$G_n$ is exactly~$2$, as otherwise, the number of edges of~$G_n$ would exceed~$2n-4$.
Let~$v_1$ be a vertex of degree~$n-1$ in~$G_n$ and let~$v_n$ be a vertex of degree~$2$ in~$G_n$.
Recall that the central cycle in~$G_n$ is an induced cycle of length 4.
This implies that~$v_1$, which has degree~$n-1$, is not part of the central cycle.
This further implies that~$v_n$ is also not part of the central cycle, since~$v_1v_n$ is an edge of~$G_n$ (due to the degree of~$v_1$) and~$v_n$ has only one other neighbor in~$G_n$.
Hence, $G' := G_n - v_n$ is a~\cp graph, since (i)~$v_n$ has degree 2 in~$G_n$ and (ii)~$v_n$ is not part of the central cycle and thus connected with exactly one edge to each of the two spanning trees.
But $G'$ cannot be a~\cp graph, due to the following.
If~$G'$ has a minimum degree of at most 1, our initial argument shows that~$G'$ is not a~\cp graph. If~$G'$ has a minimum degree of at least 2, since~$G'$ has~$n' = n-1$ vertices and a vertex of degree~$n'-1$ (namely $v_1$), the induction hypothesis implies that~$G'$ is not a~\cp graph.
This contradicts the assumption that~$G_n$ is a~\cp graph, which shows the statement for~$n$.
\end{proof}

This directly proves the forward direction of~\Cref{c4 sequences}.

\subsection{Realizability}
Now, we show the backward direction~of~\Cref{c4 sequences}.
That is, we show the following.
\begin{lemma}\label{c4 really realizable}
Let~$\dd = (d_1, \dots, d_n)$ be a graphical sequence with $\sum_{i=1}^n d_i = 4(n-1)-4$, $d_1 < n-1$, and~$d_n \geq 2$.
Then, $\dd$ admits a realization which is a~\cp graph.
\end{lemma}

To prove this statement, we assume that we are given a graphical sequence~$\dd$ with $\sum\dd = 4(n-1)-4$, $d_1 < n-1$, and~$d_n \geq 2$.
Note that these sequences are guaranteed to fulfill~$d_n \in \{2,3\}$, as otherwise, $\sum \dd \geq 4 n > 4(n-1) - 4$.
We split the analysis in three cases. 
First, we consider the case of $d_n=3$~and~$d_1\leq 4$ in~\Cref{minus four dn eq 3 d1 leq 4}, then $d_n=3$~and~$5\leq d_1 < n-1$ in~\Cref{minus four dn eq 3 d1 geq 5}, and lastly $d_n=2$~and~$d_1<n-1$ in~\Cref{minus four dn eq 2}.

\begin{lemma}\label{minus four dn eq 3 d1 leq 4}
Let~$\dd = (d_1, \dots, d_n)$ be a graphical sequence with $\sum_{i=1}^n d_i = 4(n-1) - 4$, $d_1 \leq 4$ and~$d_n = 3$. 
Then~$\dd$ admits a realization which is a~\cp graph.
\end{lemma}
\begin{proof}
Since~$\sum\dd = 4(n-1) - 4$, $d_1 \leq 4$, and~$d_n = 3$, we get that~$n\geq 8$, $d_{n-7} = 3$, and (if~$n>8$), $d_{n-8} = 4$.
In other words, for each~$n\geq 8$, there is exactly one degree sequence~$\dd^n$ of length~$n$ with the required properties.

We show that~$\dd^n$ admits a realization which is a~\cp graph via induction over~$n$.
In fact, we show a stronger result, namely, that for each~$n\geq 8$, $\dd^n$ has a realization~$G$ which is a~\cp graph with central cycle~$C$, such that (i)~there are two spanning trees~$T_1$ and~$T_2$ of~$G$ that share exactly two edges and both belong to~$C$, and (ii)~for each~$i\in [1,2]$, there are two distinct edges~$e_i,f_i$ from~$T_i$ and outside of~$C$ for which~$\{e_1,e_2\}$ and~$\{f_1,f_2\}$ are matchings in~$G$.
We call edge pairs~$(e_1,e_2)$ and~$(f_1,f_2)$ fulfilling these  conditions~\emph{matching edge pairs}.

\begin{figure}
    \centering
    \scalebox{.85}{
    \includegraphics[height=0.25\linewidth]{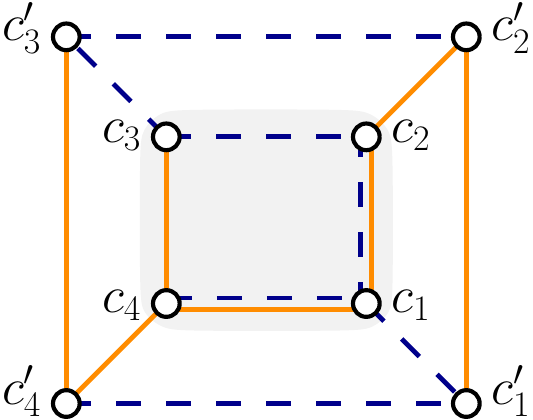}\hspace{0.1\linewidth}
    \includegraphics[height=0.25\linewidth]{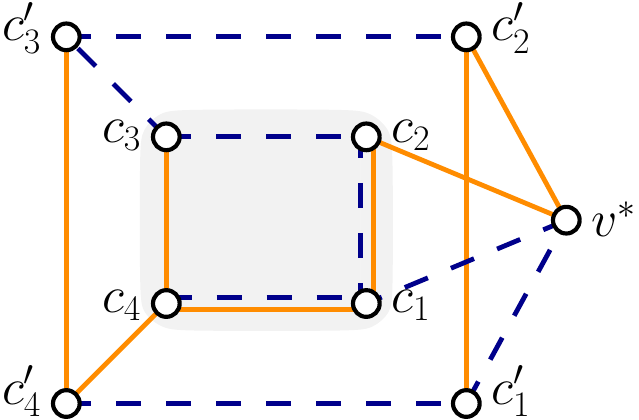}
    }
    \caption{
    Realization of the degree sequence \( \dd^8 = (3,3,3,3,3,3,3,3) \) from the base case in the proof of \Cref{minus four dn eq 3 d1 leq 4}, featuring the central 4-cycle \( C\) highlighted in gray, and two spanning trees \( T_1 \) (solid orange edges) and \( T_2 \) (dashed blue edges) that share exactly two edges, both contained in \( C \).  
    The right side illustrates the inductive step.}
    \label{fig:c4 cycle base case n8}
\end{figure}
For the base case of~$n = 8$, consider the graph given in \Cref{fig:c4 cycle base case n8} 
which is a realization of~$\dd^8 = (3,3,3,3,3,3,3,3)$.
As highlighted in the figure, there is a central cycle~$C$ on the vertices~$c_1,c_2,c_3,c_4$.
Moreover, $(c_1c_1',c_2c_2')$ and~$(c_3c_3',c_4c_4')$ are matching edge pairs.
Hence, the statement holds for the base case of~$n=8$.

Now let~$n>8$ and assume that the statement holds for~$n - 1$.
We show that the statement also holds for~$n$.
Let~$G'$ be a realization of~$\dd^{n-1}$ according to the properties of the induction hypothesis.
That is, $G'$ is a~\cp graph with central cycle~$C$, two spanning trees~$T_1'$ and~$T_2'$ that share exactly two edges, both of which belong to~$C$, and  matching edge pairs~$(e_1',e_2')$ and~$(f_1',f_2')$.
Recall that~$\dd^n$ can be obtained from~$\dd^{n-1}$ by adding a single entry of value~$4$.
The idea to obtain a realization for~$\dd^n$ with these properties is by essentially subdividing both edges~$e_1' = pq$ and~$e_2' = xy$ and identifying both newly added degree-2 vertices to obtain a new vertex~$v^*$ (which then has degree 4). 
Formally, we obtain a graph~$G$ by removing the edges~$e_1'$ and~$e_2'$ from~$G'$ and afterwards adding a new vertex~$v^*$ together with the edges~$pv^*,qv^*,xv^*,yv^*$.

\begin{claim}
$G$ is a~\cp graph realization of~$\dd^n$ that has matching edge pairs.
\end{claim}

Note that each vertex of~$V(G)\setminus \{v^*\}$ has the same degree in both~$G$ and~$G'$ and~$v^*$ has degree 4.
Hence, $G$ is a realization of~$\dd^n$.
Moreover, since~$e_1'$ and~$e_2'$ are not part of the central cycle~$C$ of~$G'$, $C$ is still a central cycle where~$T_1 := (T'_1 - e'_1) \cup \{pv^*, qv^*\}$ and~$T_2 := (T'_2 - e'_2) \cup \{xv^*, yv^*\}$ being spanning trees of~$G$ that each share exactly two edges and both of which belong to~$C$.
This implies that~$G$ is a~\cp graph.
It thus remains to show that there are matching edge pairs~$(e_1,e_2)$ and~$(f_1,f_2)$ in~$G$.
To define these pairs, we use the so far undiscussed pair~$(f_1',f_2')$.
Since~$e_1',e_2',f_1',f_2'$ are pairwise distinct edges in~$G'$, $f_1'$ and~$f_2'$ are still edges of~$G$ and moreover, $f_1'$ is an edge of~$T_1$ and~$f_2'$ is an edge of~$T_2$.
Recall that~$xv^*$ and~$yv^*$ are edges of~$T_2$, and note that~$f_1'$ contains at most one of~$x$ and~$y$ as an endpoint, as otherwise $f_1'$ would be identical to~$e_2'$.
Hence, there is an edge~$e_2\in\{xv^*,yv^*\}$ which is not adjacent to~$f_1'$, that is, $\{f_1',e_2\}$ is a matching in~$G$.
With the same arguments, there is an edge~$e_1\in\{pv^*,qv^*\}$ which is not adjacent to~$f_2'$, that is,~$\{e_1,f_2'\}$ is a matching in~$G$.
This implies that~$(f_1',e_2)$ and~$(e_1, f_2')$ are matching edge pairs.
Consequently, the statement holds for~$n$.

Concluding, for each~$n\geq 8$, the degree sequence~$\dd^n$ admits a \cp graph realization.
\end{proof}

Next, we consider the case that~$5\leq d_1 < n-1$ and~$d_n = 3$.

\begin{figure}
    \centering
    \includegraphics[width=0.5\linewidth]{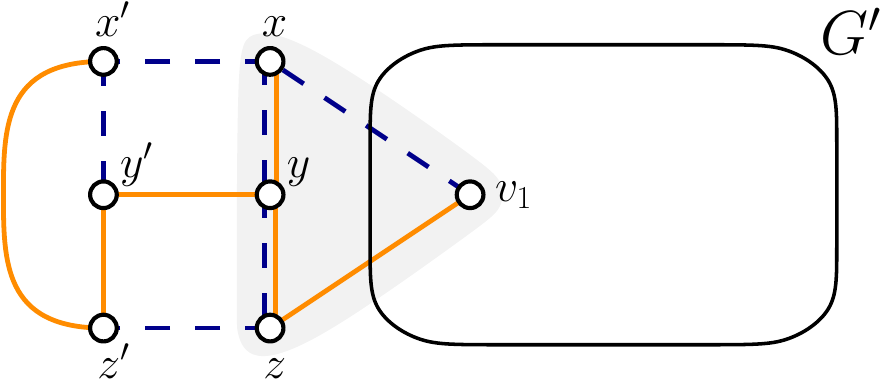}
    \caption{Construction of the graph $G$ described in the proof of \Cref{minus four dn eq 3 d1 geq 5}.}
    \label{fig:d1 greater 4 and dn equal 3}
\end{figure}

\begin{lemma}\label{minus four dn eq 3 d1 geq 5}

Let~$\dd = (d_1, \dots, d_n)$ be a graphical sequence with $\sum_{i=1}^n d_i = 4(n-1) - 4$, $5 \leq d_1 < n-1$, and~$d_n = 3$. 
Then~$\dd$ admits a realization which is a~\cp graph.
\end{lemma}

\begin{proof}
Before we show that~$\dd$ admits a realization which is a~\cp graph, we observe some structure of~$\dd$.
As we show, the condition that~$d_1 < n-1$ follows implicitly by~$\sum \dd = 4(n-1) - 4$ and~$d_n = 3$.
Completely analogous to~\Cref{claim number of threes}, we can bound the number of entries of value~$3$ in~$\dd$.

\begin{claim}\label{claim number of threes c4}
Let~$\ell$ be the smallest index of~$[1,n]$ for which~$d_\ell = 3$. 
Then, the number of entries of value~$3$ in~$\dd$ is~$8 + \sum_{i=1}^{\ell-1} (d_i - 4)$.
\end{claim}
\Cref{claim number of threes c4} implies that~$n \geq 1 + d_1 - 4 + 8 = d_1 + 5$.
Hence, $n\geq 10$ by the fact that~$d_1 \geq 5$. 
Further, we get that~$d_1 \leq n - 5$.
We can further prove that~$d_2 \leq n-7$.
 To this end, assume towards a contradiction that~$d_2 \geq n-6$.
Since~$n\geq 10$, this implies that~$d_2 \geq 4$.
Note that this implies~$n\geq 11$, since~$\dd$ contains at least~$8 + 1 = 9$ entries of value~$3$ by~\Cref{claim number of threes c4} and the fact that~$d_1 - 4 \geq 1$.
Then, $\sum\dd = d_1 + d_2 + \sum_{i=3}^nd_i \geq 2(n-6) + (n-2)\cdot 3 = 4(n-1) - 14 + n$.
Since~$n \geq 11$, this implies that $\sum\dd \geq 4(n-1) - 3$, which is a contradiction.
Consequently, $d_1\leq n-5$ and~$d_2 \leq n-7$.
We are now ready to prove that~$\dd$ admits a realization which is a~\cp graph.

Let~$\dd':=(d_1',\dots,d'_{n'})$ with~$n'=n-6 \geq 4$ be the degree sequence obtained from~$\dd$ by removing six entries of value~$3$ and decreasing~$d_1$ by~$2$.
That is, $d_1' = \max(d_1-2,d_2) \leq n-7 \leq n'-1$ by our previous observations, and~$d'_{n'} \geq 3$, since~$d_n \geq 3$ and~$d_1 \geq 5$.
Moreover, $\sum\dd' = \sum \dd - 2 - 6\cdot 3 = 4(n-1)-4 - 20 = 4(n'-1)$.
Hence, \Cref{realizable if dn geq 3} implies that~$\dd'$ is graphical.
Moreover, \Cref{thm:edt_76} implies that there is a realization~$G'$ of~$\dd'$ that contains two edge-disjoint spanning trees~$T_1'$ and~$T_2'$.
We construct a graph~$G$ as follows (see~\Cref{fig:d1 greater 4 and dn equal 3}):
Let~$v_1$ be an arbitrary vertex of degree~$d_1-2$ in~$G'$.
We add a cycle~$(x,y,z,z',y',x')$ to~$G'$ as well as the edges~$v_1x, v_1z, yy', x'z'$.
Note that~$G$ is a realization of~$\dd$, since all six vertices~$x,y,z,x',y',z'$ have degree~$3$ in~$G$ and the degree of~$v_1$ was increased by~$2$.
Moreover, as depicted in~\Cref{fig:d1 greater 4 and dn equal 3}, $G[\{v_1,x,y,z,x',y',z'\}]$ is a~\cp graph with central cycle~$C:=G[\{v_1,x,y,z\}]$ and two spanning trees~$T^c_1$ and~$T^c_2$ that share exactly two edges, both of which belong to~$C$.
Consequently, $C$ is a central cycle in $G$, and~$T_1 := T'_1 \cup T^c_1$ and~$T_2 := T'_2 \cup T^c_2$ are two spanning trees of~$G$ that share exactly two edges, both of which belong to~$C$.
Thus, $G$ is a~\cp graph, which completes the proof.
\end{proof}

Finally, we consider the case where~$d_1 < n-1$ and~$d_n = 2$.
This proof is similar to the proof of~\Cref{minus two dn eq 2}.
Essentially, we perform an induction over the length of the sequence and lay off the smallest degree~$d_n = 2$, until we reach a sequence of constant length or~$d_n$ has value 3.
In the latter case, one of the previous two lemmas acts as a base case.

\begin{lemma}
\label{minus four dn eq 2}
Let~$\dd = (d_1, \dots, d_n)$ be a graphical sequence with $\sum_{i=1}^n d_i = 4(n-1) - 4$, $d_1 < n-1$, and~$d_n = 2$. 
Then~$\dd$ admits a realization which is a~\cp graph.
\end{lemma}

\begin{proof}
To prove the statement, we show that for each~$n$, each graphical sequence~$\dd^n$ of length~$n$ with $\sum \dd^n = 4(n-1) - 4$, $d_1 < n-1$, and~$d_n = 2$ admits a realization which is a~\cp graph.
We show this statement via induction over~$n$.

For the base case, note that for~$n \in \{1,2,3\}$, there is no graphical sequence of length~$n$ that sums up to~$4(n-1)-4$. 
Moreover, for~$n = 4$, $(2,2,2,2)$ is the only graphical sequence of length~$n$ with~$d_1< n-1 = 3$ and~$d_n = 2$.
Since the unique non-isomorphic realization of~$\dd$ is a cycle of length 4, which obviously is a~\cp graph, the base case holds for~$n\leq 4$.

For the inductive step, let~$n \geq 5$ and assume that the statement holds for~$n-1$. 
Let~$\dd^n$ be an arbitrary graphical sequence of length~$n$ fulfilling~$\sum\dd^n = 4(n-1)-4$, $d_1 < n-1$, and~$d_n = 2$.
We show that~$\dd^n$ admits a realization which is a~\cp graph.
First, observe that~$d_3 < n-2$, as otherwise~$\sum\dd = d_1+d_2+d_3 + \sum_{i=4}^n d_i \geq 3(n-2) + (n-3)\cdot 2 = 4(n-1)-8+n$, which is larger than~$4(n-1)-4$ because of~$n \geq 5$.
Since~$\dd^n$ is graphical, the residual sequence~$\dd' = (d_1', \dots, d_{n'}')$ obtained from laying off~$d_n$ is graphical (see~\Cref{lem:havel-hakimi-correctness}).
Moreover, $\dd'$ has length~$n' = n-1$ and fulfills~$\sum \dd' = \sum \dd - 4 = 4(n-1)-4 - 4 = 4(n'-1)-4$, since $d_n=2$.

Hence, $\dd'$ admits a realization which is a~\cp graph 
\begin{itemize}
\item by the induction hypothesis if~$d'_{n'} = 2$, since~$d'_1 = \max(d_1-1, d_3) < n-2 = n'-1$,
\item by \Cref{minus four dn eq 3 d1 leq 4} if~$d'_{n'} = 3$ and~$d'_{1} \leq 4$, and
\item by \Cref{minus four dn eq 3 d1 geq 5} if~$d'_{n'} = 3$ and~$d'_{1} \geq 5$.
\end{itemize}
Note that~$d'_{n'} < 2$ is not possible, since~$d_2 > 2$.
The latter follows from the fact that $\sum \dd^n = 4(n-1)-4$ and~$d_{n} = 2$ for~$n\geq 5$.
Now, let~$G'$ be a realization of~$\dd'$ which is a~\cp graph, let~$v_1$ be a vertex of degree~$d_1-1$ in~$G'$, and let~$v_2$ be a vertex of degree~$d_2-1$ in~$G'$.
These vertices exists, since~$G'$ is a realization of~$\dd'$ and~$\dd'$ is obtained from~$\dd^n$ by laying off~$d_n=2$.
Consider the graph~$G$ obtained by adding a vertex~$v_n$ to~$G'$ as well as the edges~$v_1v_n$ and~$v_2v_n$.
Note that~$G$ is a realization of~$\dd^n$.
Since~$G'$ is a~\cp graph, $G'$ contains a central cycle~$C$ and two spanning trees~$T'_1$ and~$T'_2$, such that~$T_1'$ and~$T_2'$ share exactly two edges and the two shared edges are from~$C$.
Hence, $C$ is a central cycle in~$G$, and $T_1 := T'_1 \cup v_1v_n$ and~$T_2 := T'_2 \cup v_2v_n$ are spanning trees of~$G$, such that~$T_1$ and~$T_2$ share exactly two edges and the two shared edges are from~$C$.
This implies that~$G$ is a~\cp graph.
Since~$\dd^n$ is an arbitrary graphical sequence of length~$n$ fulfilling~$\sum\dd^n = 4(n-1)-2$, $d_1 < n-1$, and~$d_n = 2$, the statement holds for~$n$.
\end{proof}

\section{The Characterizations for TC-realizable Sequences}\label{sec:characterizations}

We now prove the complete characterization for graphical and multigraphical TC-realizable sequences.

\subsection{Realizations of TC-realizable Graphs}
We now prove our characterization of TC-realizable graphical sequences:
\begin{theorem}\label{main gra}
A graphical sequence~$\dd = (d_1, \dots, d_n)$ with $s=\sum_{i=1}^{n} d_i$ is TC-realizable if and only if one of the following holds
\begin{itemize}
\item $s = 4(n-1)-4$, $d_1 < n-1$, and~$d_n \geq 2$
\item $s \geq 4(n-1)-2$ and (a)~$n \leq 2$ or (b)~$n > 2$, $d_{n-1} \geq 2$, and~$d_n \geq 1$.
\end{itemize}
\end{theorem}
\begin{proof}
$(\Leftarrow)$
If~$\sum \dd = 4(n-1)-4$, $d_1 < n-1$, and~$d_n \geq 2$, then due to~\Cref{c4 sequences}, $\dd$ has a realization~$G$ which is a~\cp graph.
By Harary~and~Schwenk~\cite{harary_communication_1974} this implies that~$G$ is TC-realizable and thus that~$\dd$ is TC-realizable.

If~$\sum \dd \geq 4(n-1)-2$ and (i)~$n \leq 2$ or (ii)~$n > 2$, $d_{n-1} \geq 2$, and~$d_n \geq 1$, then $\dd$ admits a realization~$G$ which has two spanning trees that share at most one edge (see~\Cref{two spanning trees}).
By Harary~and~Schwenk~\cite{harary_communication_1974} this implies that~$G$ is TC-realizable and thus that~$\dd$ is TC-realizable.

$(\Rightarrow)$
Suppose that~$\dd$ is TC-realizable and let~$G$ be one of its realizations.
By Hajnal~et~al.~\cite{hajnal_cure_1972}, $G$ has at least~$2n-4$ edges, which implies~$\sum\dd \geq 4(n-1)-4$.

First, consider the case that~$\sum\dd = 4(n-1)-4$.
That is, $G$ has~$2n-4$ edges.
Hence, by Göbel et al.~\cite{gobel_label-connected_1991}, $G$ is a~\cp graph.
Thus, by~\Cref{c4 sequences}, $d_1 < n-1$, and~$d_n \geq 2$.

Second, consider the case that~$\sum\dd \geq 4(n-1)-2$.
Since~$G$ is TC-realizable, $G$ is connected.
For~$n\leq 2$, $\dd$ is from~$\{(0),(1,1)\}$ as these are the only graphical sequences of length at most~2 fulfilling~$\sum\dd\geq 4(n-1)-2$.
Thus, consider~$n > 2$.
Since~$G$ is connected and~$n > 2$, $G$ has a minimum degree of 1, which implies that~$d_n \geq 1$.
It remains to show that~$d_{n-1}\geq 2$.
To this end, assume towards a contradiction that~$d_{n-1} = d_n = 1$.
Then, $G$ contains at least two vertices~$u$ and~$v$ of degree 1.
These vertices are not adjacent in~$G$, as otherwise, $G$ would not be connected.
Hence, the unique edge~$e_u$ incident with~$u$ and the unique edge~$e_v$ incident with~$v$ are distinct.
Since~$G$ is TC-realizable, there is a TC-labeling of the edges of~$G$.
To ensure that there is a temporal path from~$u$ to~$v$, this labeling has to assign a label to $e_u$ which is strictly smaller than the label assigned to~$e_v$.
But simultaneously, to ensure that there is a temporal path from~$v$ to~$u$, the label assigned to $e_v$ has to be strictly smaller than the label assigned to~$e_u$; a contradiction.
Consequently, $G$ contains no two vertices of degree~$1$, which implies that~$d_{n-1} \geq 2$ and~$d_n \geq 1$ in the case that~$n > 2$.
\end{proof}

\subsection{Realizations of TC-realizable Multigraphs}\label{sec:multi}

We now prove our characterization for TC-realizable multigraphical sequences:

\begin{theorem}\label{main multi}

A multigraphical sequence~$\dd = (d_1, \dots, d_n)$ with $s:= \sum_{i=1}^n d_i$ is TC-realizable if and only if one of the following holds
\begin{itemize}
\item $s = 4(n-1)-4$ and $d_n \geq 2$
\item $s \geq 4(n-1)-2$ and (a)~$n \leq 2$ or (b)~$n > 2$, $d_{n-1} \geq 2$, and~$d_n \geq 1$.
\end{itemize}
\end{theorem}
Note that the only difference from the characterization of (simple) graphical sequences (see \Cref{main gra}) is the absence of the restriction $d_1<n-1$ in the case $m=4(n-1)-4$.

To prove the statement, we again consider both items independently and mostly use the same arguments as for the graphical case.

\subsection{Intermediate characterization for spanning trees with one shared edge}
First, we provide a characterization for multigraphical sequences that allow for a realization that has two spanning trees with at most one shared edge.

\begin{theorem}

\label{two spanning trees multi}
Let~$\dd = (d_1, \dots, d_n)$ be a multigraphical sequence.
Then, $\dd$ admits a realization with two spanning trees that share at most one edge if and only if
$\sum_{i=1}^n d_i \geq 4(n-1)-2$ and (a)~$n\leq 2$ or (b)~$n > 2$, $d_{n-1} \geq 2$, and~$d_n \geq 1$.
\end{theorem}

\begin{proof}
$(\Rightarrow)$
Suppose that~$\dd$ admits a realization~$G$ with two spanning trees that share at most one edge.
For the (multi)graph~$G$ to have two spanning trees that share at most one edge, $G$ needs at least~$2(n-1)-1$ edges.
This implies $\sum_{i=1}^n d_i \geq 4(n-1)-2$.
Further, by the same arguments used to prove~\Cref{no two trees if small degree}, we get that either~$n\leq 2$ or, if~$n > 2$, $d_{n-1} \geq 2$ and~$d_n \geq 1$.
This was due to the fact that if any of these conditions would not hold, then each realization of~$\dd$ is disconnected or contains at least two distinct edges incident with degree-1 vertices; hence, no realization with two spanning trees that share at most one edge is possible.
This completes the proof of this direction.

$(\Leftarrow)$
First consider the case of~$n \leq 2$.
If~$n = 1$, then~$(0)$ is the only multigraphical sequence, for which the unique realization has two spanning trees that share no edge.
If~$n=2$, then each multigraphical sequence~$(d_1,d_2)$ fulfills~$d_1 = d_2$. Out of all such sequences, $(0,0)$ is the only that does not fulfill $\sum\dd\geq 4(n-1)-2$. Every other such sequence has a unique realization and this realization is connected.
Since~$n = 2$, the connectivity implies that there are two spanning trees, both consisting of one multiedge, that share at most one edge.

In the remainder, consider~$n > 2$ and assume that~$d_{n-1}\geq 2$ and~$d_n \geq 1$.
First, consider~$d_n = 1$.
Since~$n>2$ and~$\sum\dd \geq 4(n-1) -2$, $d_1\geq 3$.
Then, by laying off~$d_n$, the resulting degree sequence~$\dd'$ is also multigraphical due to~\Cref{laying off multi}.
Recall that laying off~$d_n$ means, that we remove~$d_n$ from~$\dd$ and decrement~$d_1$ by~$d_n=1$ (and potentially reorder the degree sequence).
The resulting degree sequence~$\dd'$ has then length~$n':=n-1$ and fulfills~$\sum\dd' = \sum\dd -2 \geq 4(n-1)-2-2 = 4(n'-1)$.
Hence, due to~\Cref{two spanning trees multi}, $\dd'$ has a realization~$G'$ that has two disjoint spanning trees~$T_1'$ and~$T_2'$.
Let~$v_1$ be an arbitrary vertex of~$G'$ of degree~$d_1-1$.
Then, the multigraph~$G$ obtained by adding a new vertex~$v_n$ to~$G'$ and the edge~$v_1v_n$ is a realization of~$\dd$.
Moreover, $G$ has the two spanning trees~$T_1:=T'_1 \cup v_1v_n$ and~$T_2:=T'_2 \cup v_1v_n$ that share only one edge.
Hence, the statement holds if~$d_n = 1$.

Thus, consider~$d_n \geq 2$.
If~$\sum\dd \geq 4(n-1)$, then~\Cref{two spanning trees multi sequence} implies that~$\dd$ has a realization with two disjoint spanning trees, which shows the statement under this assumptions.
Hence, it remains to consider the case that~$\sum\dd = 4(n-1)-2$.
Note that this implies that~$d_n \in \{2,3\}$.
We distinguish between the values of~$d_n$.

If~$d_n = 3$, then $n\geq 6$ and~$d_1 \leq n-6$ since~$\sum\dd = 4(n-1) -2$.
By~\Cref{realizable if dn geq 3} this implies that~$\dd$ is also a graphical sequence, for which the statement holds due to~\Cref{minus two dn eq 3}.

It thus remains to consider the case that~$d_n = 2$.
To this end, note that since~$d_n = 2$, $\sum_{i=2}^n d_i \geq 2(n-1)$.
This implies~$d_1 \leq \sum\dd - 2(n-1) = 4(n-1)-2 -2(n-1) = 2(n-1)-2$.
Consider the degree sequence~$\dd'$ obtained by removing~$d_n$ from~$\dd$.
Then, $\dd'$ has length~$n':= n-1$ and fulfills~$\sum\dd' = \sum\dd - 2 = 4(n-1)-4 = 4(n'-1)$.
Since~$d_1 \leq 2(n-1)-2 = 2(n'-1) \leq \frac{1}{2}\cdot \sum\dd'$, \Cref{multi realize chara} implies that~$\dd'$ is multigraphical.
Moreover, since~$d_{n-1} \geq  d_n = 2$, $\dd'$ has a realization~$G'$ with two disjoint spanning trees~$T_1'$ and~$T_2'$.
Let~$uv$ be an arbitrary edge of~$T_1'$ and let~$G$ be the graph obtained from~$G'$ by subdividing edge~$uv$ into a new vertex~$v_n$.
The multigraph~$G$ is a realization for~$\dd$, since each vertex of~$G'$ has the same degree in~$G$, and~$v_n$ has degree~2.
Moreover, $G$ has the two spanning trees~$T_1:= (T_1' - uv) \cup \{uv_n,vv_n\}$ and~$T_2:= T_2' \cup uv_n$ that share only one edge (namely~$uv_n$).
Hence, the statement also holds for~$d_n=2$ and thus in all stated cases.
\end{proof}

\subsection{Intermediate characterization for~\cp multigraphs}
Next, we provide a characterization for multigraphical sequences that allow for a realization which is a~\cp multigraph.

\begin{theorem}\label{c4 sequences multi}
Let~$\dd = (d_1, \dots, d_n)$ be a multigraphical sequence.
Then, $\dd$ admits a realization which is a~\cp multigraph if and only if
$\sum_{i=1}^n d_i = 4(n-1)-4$ and~$d_n \geq 2$.
\end{theorem}
\begin{proof}
$(\Rightarrow)$
By definition, each~\cp multigraph on~$n$ vertices has exactly~$2n-4$ multiedges.
Hence, $\sum\dd = 4(n-1)-4$ holds.
Next, we show that~$d_n \geq 2$ also has to hold.
The argument is identical to the one in the proof of~\Cref{c4 degree bounds}:
By definition, each vertex~$v$ of the~\cp multigraph~$G$ is (i)~part of the central cycle or (ii)~a vertex that has at least one incident multiedge in each of the two spanning trees that only share edges of the central cycle.
In both cases, vertex~$v$ has degree at least 2.
Hence, the minimum degree of~$G$ is at least 2, which implies that~$d_n \geq 2$ to allow for a realization which is a~\cp multigraph.
This completes this direction.

$(\Leftarrow)$
We have to show that~$\dd$ has a realization which is a~\cp multigraph if~$\sum\dd = 4(n-1)-4$ and~$d_n \geq 2$.
Note that this implies that~$n\geq 4$.
To show the statement, we distinguish between the possible values of~$d_{n-2}$.

If~$d_{n-2}\geq 3$, then~$\sum\dd - d_1 \geq 3(n-3) + d_{n-1} + d_n \geq 3(n-3) + 4 = 3(n-1) - 2$.
By~$\sum\dd = 4(n-1)-4$, this implies~$d_1 \leq n-3$.
Thus, by~\Cref{realizable if dn geq 3}, $\dd$ is also a graphical sequence, since~$d_4 \geq 3$.
This is due to the fact that~$d_{n-2} \geq 3$ and~$n \geq 8$, as otherwise, $\sum \dd \geq 3n > 4(n-1)-4$.
Hence, \Cref{c4 sequences} thus implies that~$ \dd$ admits a realization which is a~\cp graph (and thus a~\cp multigraph).

Hence, we only need to consider the case that~$d_{n-2} = 2 = d_{n-1} = d_{n}$.
First, we consider special cases: $\dd \in \{(2,2,2,2), (3,3,2,2,2)\}$.
Note that these are the only two options for~$\dd$ for which~$\sum\dd = 4(n-1)-4$, $d_1 \leq 4$, and~$d_{n-2} = d_{n-1} = d_{n} = 2$. Both these degree sequences are graphical and fulfill~$d_1 < n-1$ and~$d_n \geq 2$.
Hence, by~\Cref{c4 sequences}, $\dd$ admits a realization which is a~\cp graph (and thus a~\cp multigraph).

For all other options of~$\dd$, we can thus assume that~$d_1 \geq 4$.

We upper-bound the values of~$d_1$ and~$d_2$.
Since~$d_n = 2$, $\sum \dd - d_1 \geq 2(n-1)$.
Moreover, by~$\sum\dd = 4(n-1)-4$, this implies that~$d_1 \leq 2(n-1)-4$.
Furthermore, $d_2 \leq 2(n-1) - 6$, as otherwise, $d_1+d_2 > 4(n-1) - 10$, which would imply~$\sum\dd \geq d_1 + d_2 + (n-3)\cdot 2 > 4(n-1) - 10 + 2(n-3) \geq 4(n-1)-4$ by~$n\geq 5$.
With these upper bounds for~$d_1$ and~$d_2$, we can now present a realization of~$\dd$ which is a~\cp multigraph.

Let~$\dd' := (d_1', \dots, d_{n'}')$ be the degree sequence with~$n' := n-3$ obtained from~$\dd$ by (i)~removing~$d_{n-2}$, $d_{n-1}$, and~$d_{n}$, and by (ii)~reducing~$d_1$ by~$2$ (and potentially reordering).

Hence, $d_1' = \max(d_1-2, d_2) \leq 2(n-1)-6 = 2(n'-1)$ by the above argumentation.
Moreover, $\dd'$ fulfills~$\sum\dd' = \sum\dd - 8 = 4(n-1)-4 - 8 = 4(n'-1)$.
Thus, $d'_1 \leq \sum\dd' -d_1'$, which implies that~$\dd'$ is multigraphical by~\Cref{multi realize chara}.
Furthermore, $d_{n'}' = \min(d_1-2, d_{n-3}) \geq 2$ by the fact that~$n > 4$, $d_n \geq 2$, and~$d_1 \geq 4$.
\Cref{two spanning trees multi} thus implies that~$\dd'$ has a realization~$G'$ which has two disjoint spanning trees~$T_1'$ and~$T_2'$.

We obtain a realization~$G$ of~$\dd$ as follows:
Let~$v_1$ denote an arbitrary vertex of~$G'$ with degree~$d_1-2$.
We add three new vertices~$x,y$, and~$z$ to~$G'$ and add the edges~$xy,yz,v_1x,v_1z$.
Note that~$G$ is a realization of~$\dd$, since we added three new vertices of degree 2 to~$G'$ and increased the degree of a vertex of degree~$d_1-2$ by~$2$.
It remains to show that~$G$ is a~\cp multigraph.
This is the case, since~$C:=(v_1,x,y,z)$ is an induced cycle of length 4 in~$G$ and~$T_1 := T_1' \cup \{v_1x,xy,yz\}$ and~$T_2:= T_2' \cup \{v_1z,xy,yz\}$ are two spanning trees of~$G$ that only share two edges and both these edges are from~$C$ (namely the edges~$xy$ and~$yz$).
Consequently, $\dd$ admits a realization which is a~\cp multigraph.
\end{proof}

With these characterizations, we can prove our characterization for TC-realizable multigraphical sequences, namely~\Cref{main multi}.

\begin{proof}[Proof of~\Cref{main multi}]
The proof is very similar to the one of~\Cref{main gra}. 

$(\Leftarrow)$
If~$s = 4(n-1)-4$ and~$d_n \geq 2$, then due to~\Cref{c4 sequences multi}, $\dd$ has a realization~$G$ which is a~\cp multigraph.
By~\Cref{c4 property multi} this implies that~$G$ is TC-realizable and thus that~$\dd$ is TC-realizable.

If~$s \geq 4(n-1)-2$ and (a)~$n \leq 2$ or (b)~$n > 2$, $d_{n-1} \geq 2$, and~$d_n \geq 1$, then $\dd$ admits a realization~$G$ which has two spanning trees that share at most one edge (see~\Cref{two spanning trees multi}).
By~\Cref{one edge towards EDSTs minus one} this implies that~$G$ is TC-realizable and thus that~$\dd$ is TC-realizable.

$(\Rightarrow)$
Suppose that~$\dd$ is TC-realizable and let~$G$ be one of its realizations.
By Hajnal~et~al.~\cite{hajnal_cure_1972}, $G$ has at least~$2n-4$ edges, which implies~$\sum\dd \geq 4(n-1)-4$.

Firstly, consider the case that~$\sum\dd = 4(n-1)-4$.
We need to show that~$d_n \geq 2$.
Note that~$G$ has exactly~$2n-4$ edges.
Assume towards a contradiction that~$d_n \leq 1$.
Then, $G$ contains at least one vertex~$v_n$ of degree at most~$1$.
If~$v_n$ has degree~$0$, $G$ is not connected; a contradiction to the fact that~$G$ is TC-realizable.
Otherwise, there exists a single edge~$e_n$ incident with~$v_n$.
Let~$\lambda\colon E \to \mathbb{N}$ be a TC labeling of~$G$ and let~$v^*$ be the unique neighbor of~$v_n$ in~$G$.
Since~$e_n$ is the only edge incident with~$v_n$, all temporal paths from and towards~$v_n$ traverse edge~$e_n$ at time~$\lambda(e_n)$.
Moreover, since each vertex of~$V(G)\setminus \{v_n\}$ has a temporal path from and to~$v_n$, this implies that for each vertex~$w$ of~$V(G)\setminus \{v_n\}$ there is (i)~a temporal path from~$w$ that reaches~$v^*$ prior to label~$\lambda(e_n)$ and (ii)~a temporal path from~$v^*$ to~$w$ that starts later than label~$\lambda(e_n)$.
This means in particular that~$G - v_n$ has at least two edge-disjoint spanning trees.
Hence, $G - v_n$ contains at least~$2(n-1) = 2n-2$ edges.
This contradicts the fact that~$G$ has exactly~$2n-4$ edges.
Consequently, $d_n \geq 2$.

Secondly, consider the case that~$\sum\dd \geq 4(n-1)-2$.
Since~$G$ is TC-realizable, $G$ is connected.
For~$n\leq 2$, all multigraphical sequences have connected realizations apart from $(0,0)$, which must not be considered as it does not fulfill~$\sum\dd\geq 4(n-1)-2$.
Hence, consider~$n > 2$.
Since~$G$ is connected and~$n > 2$, $G$ has a minimum degree of 1, which implies that~$d_n \geq 1$.
It remains to show that~$d_{n-1}\geq 2$.
To this end, assume towards a contradiction that~$d_{n-1} = d_n = 1$.
Then, $G$ contains at least two vertices~$u$ and~$v$ of degree 1.
These vertices are not adjacent in~$G$, as otherwise, $G$ would not be connected.
Hence, the unique edge~$e_u$ incident with~$u$ and the unique edge~$e_v$ incident with~$v$ are distinct.
Since~$G$ is TC-realizable, there is a TC labeling of the edges of~$G$.
To ensure that there is a temporal path from~$u$ to~$v$, this labeling has to assign a label to $e_u$ which is strictly smaller than the label assigned to~$e_v$.
But simultaneously, to ensure that there is a temporal path from~$v$ to~$u$, the label assigned to $e_v$ has to be strictly smaller than the label assigned to~$e_u$; a contradiction.
Consequently, $G$ contains no two vertices of degree~$1$, which implies that~$d_{n-1} \geq 2$ and~$d_n \geq 1$ in the case that~$n > 2$.
\end{proof}

\section{Linear-Time Algorithm for Realizing TC (Multi)graphs}\label{sec:algos}
In the previous section, we presented a characterization of (multi)graphical degree sequences that are TC-realizable (see~\Cref{main gra,main multi}).
These necessary and sufficient conditions can be checked in $\Oh(n)$~time in the graphical  case and in $\Oh(n+m)$~time in the multigraphical case, where~$m:= \frac{1}{2}\cdot \sum\dd$ denotes the number of edges in each realization.
Hence, in linear time, we can detect whether there is a TC-realizable graph for the given degree sequence.
In this section, we go one step further and show that in both the graphical and multigraphical case, we can construct such a realization (together with a TC labeling) in $\Oh(n+m)$~time.
To this end, we observe that all existential proofs in the previous sections are constructive and allow for linear time implementations with the help of two simple data structures.

\subsection{The data structures}
In the following, we assume that the graphs are stored via adjacency lists.
Since our arguments rely on spanning trees, we assign two Boolean flags to each edge to indicate whether this edge belongs to the first tree, the second tree, both trees, or neither trees.
This allows us to extract the two spanning trees in $\Oh(n+m)$~time, once the final realization is constructed.

The data structure for degree sequences is as follows.
Instead of storing a degree sequence as a sequence of numbers, we store it as a double-linked list that stores entries of the form~$(x,c)$.
Such entries indicate that the respective degree sequence contains exactly~$c$ entries of value~$x$.
The entries of the double-linked list are sorted decreasingly according to their~$x$-values.
In this way, for both the graphical and multigraphical cases, we can lay off a degree~$d_i$ in $\Oh(d_i)$~time, since we only need to remove the~$d_i$ entry and rearrange the~$c$-values of at most~$d_i+1$ previous entries of the data structure.
Using this data structure, we can ensure that the degree sequence always stays sorted and no reordering is necessary.

Similarly, the data structure for (multi)graphs~$G$ is as follows.
It is also a double-linked list and stores entries of the form~$(x,S)$.
Such entries indicate that the vertices of~$S$ are exactly the vertices of~$G$ that have degree exactly~$x$.
These entries are also sorted decreasingly according to their~$x$-values. 
This way, we can reattach a vertex of degree~$d_n$ to~$G$ by adding edges properly to vertices of the first~$d_n$ entries of the data structure.
That is, we can obtain a realization for the degree sequence that, after laying off a vertex of degree~$d_n$, leads to the degree sequence of~$G$.
In particular, this can be done in $\Oh(d_n)$~time.

In case of a realization for a~\cp (multi)graph, we also output the central cycle separately.

\subsection{Realizations for two edge-disjoint spanning trees}
    We now turn to the algorithmic aspects of finding a realization with two edge-disjoint spanning trees and show that such a realization can be computed in linear time.
    
    For the graphical case, we used the result by Kundu~\cite{kundu_disjoint_1974}.     While their proof is constructive in nature, the arguments are brief and incomplete.
    Thus, we give a complete construction in \Cref{lem: simple constructing 2 edst in poly} and argue how it can be implemented in $\Oh(n+m)$ time using our data structure.
    
    For the multigraphical case, we build on the result by Gu et al.~\cite{gu_multigraphic_2012}, which was proven non-constructively. 
    In \Cref{lem: multi constructing 2 edst in poly}, we give a complete inductive construction for such multigraphical sequences, and argue that this construction, too, can be carried out in $\Oh(n+m)$ time using our data structures.

    \begin{lemma}\label{lem: simple constructing 2 edst in poly}
    
        Let $\dd=(d_1,\dots,d_n)$ be a graphical sequence with $\sum\dd\geq4(n-1)$ and $d_n\geq 2$. 
        We can compute a realization~$G$ of $\dd$ and two disjoint spanning trees of~$G$ in $\Oh(n+m)$~time.
    \end{lemma}
    \begin{proof}
        To prove the statement, we show that for each $n$ and for each graphical sequence $\dd$ of length $n$ with $\sum\dd\geq4(n-1)$ and $d_n\geq 2$, we can construct a realization with two edge-disjoint spanning trees. We show this statement via induction over $n$.
        
        For the base case, note that for~$n \in \{1,2,3\}$, there is no graphical sequence of length~$n$ that sums up to~$4(n-1)$. 
        Moreover, for~$n = 4$, $(3,3,3,3)$ is the only graphical sequence of length~$n$ with~$\sum\dd\geq 4(n-1)$ and~$d_n \geq 2$.
        Since the complete graph on four vertices, which obviously contains two edge-disjoint spanning trees, is the only realization of~$\dd$, the base case holds for~$n\leq 4$.
        
        For the inductive step, let $n\geq 5$ and assume that the statement holds for $n-1$.
        Let $\dd$ be an arbitrary graphical sequence of length $n$ fulfilling $\sum\dd\geq4(n-1)$ and $d_n\geq 2$. We show that $\dd$ admits a realization with two edge-disjoint spanning trees.
        Observe that whenever $\sum\dd=4(n-1)$, it holds that $d_n\in\{2,3\}$ as otherwise $\sum\dd\geq4n$.
                 We distinguish two cases depending on whether $\sum\dd=4(n-1)$ and $d_n=3$ or not, and begin by addressing the latter. 

        \medskip\noindent{\sffamily\bfseries\color{darkgray} Case 1:}\quad 
            If $\sum\dd= 4(n-1)$ and $d_n=2$, or $\sum\dd> 4(n-1)$ and $d_n\geq 2$, consider the residual sequence $\dd' = (d_1', \dots, d_{n'}')$ of length $n'=n-1$ obtained from laying off the entry~$d_n$ according to \Cref{def:havel hakimi laying off simple}.

            Due to \Cref{lem:havel-hakimi-correctness}, $\dd'$ is graphical.
            To show that $\dd'$ admits a realization containing two edge-disjoint spanning trees, we have to show that $d_{n'}'\geq 2$ and $\sum\dd'\geq 4(n-2)$. We distinguish by the value of $d_n$.
        
            If $d_n=2$, then $d_2>2$ as otherwise $\sum\dd=d_1+2(n-1)<3(n-1)$. Hence, $d'_{n'}=\min(d_2-1,d_{n-1})\geq 2$ since $n\geq 5$, and $\sum\dd'=\sum\dd-4\geq4(n-1)-4=4(n-2)$.
                        If~$d_n=3$, then $d'_{n'}=\min(d_{n-1},d_3-1)\geq 2$. Recall that~$\sum \dd$ is even and larger than~$4(n-1)$. Hence $\sum\dd \geq 4(n-1)+2$, and it follows 
            $\sum\dd'=\sum\dd-6\geq4(n-1)+2-6=4(n-2)$.
                        If $d_n\geq4$, then $d'_{n'}\geq d_{n-1}-1\geq4-1=3$, and             it follows from $\sum\dd\geq n\cdot d_n$ that $\sum\dd'=\sum\dd-2d_n\geq n\cdot d_n - 2d_n = (n-2)\cdot d_n \geq 4(n-2)$.
        
            In all three cases, $\dd'$ admits a realization $G'$ containing two edge-disjoint spanning trees $T_1'$ and $T_2'$ by the induction hypothesis.
            Now consider the graph $G$ obtained from $G'$ by adding a vertex $v_n$ of degree $d_n$ and connecting it to $d_n$ distinct vertices $v_i$ of $G'$ of degree $d_i-1$ for $1\leq i\leq d_n$. Then $G$ is a realization of $\dd$. Moreover, $G$ contains the edge-disjoint spanning trees $T_1:=T'_1\cup v_1v_n$ and $T_2:=T'_2\cup v_2v_n$.

        \medskip\noindent{\sffamily\bfseries\color{darkgray} Case 2:}\quad         
            Otherwise, $\sum\dd= 4(n-1)$ and $d_n= 3$. In that case, consider the sequence $\dd'$ obtained from \dd by reducing $d_1$ by 1 and removing the entry $d_n$. 
            
            We show that $\dd'$ is graphical and admits a realization with two edge-disjoint spanning trees.
            Note that $\sum\dd'=4(n-1)-1-3=4(n-2)$ is even and $\dd'$ has length $n'=n-1\geq 4$. We distinguish by the value of $n'$.
            
            If $n=5$, then by $\sum\dd= 4(n-1)$, $d_n \geq 3$ and $d_1\leq n-1=4$, it follows that $\dd=(4,3,3,3,3)$ and $\dd'=(3,3,3,3)$. We showed in the base case that $\dd'=(3,3,3,3)$ is graphical and admits a realization $G'$ with two edge-disjoint spanning trees $T_1'$ and $T_2'$.
            
            If $n>5$, observe that $d'_{n'} \geq 3$ and $d'_1 =\max(d_1-1,d_2) \leq n'-1$ (otherwise $d_1\geq n$ or $\sum\dd=2(n-1)+3(n-2)>4(n-1)$). Furthermore, $d'_4\geq 3$ by $d_n\geq 3$ and $n> 5$. 
            It now follows from \Cref{realizable if dn geq 3} that $\dd'$ is graphical. Thus, $\dd'$ admits a realization $G'$ with two edge-disjoint spanning trees $T_1'$ and $T_2'$ by the induction hypothesis.

            Now consider the graph $G$ obtained from $G'$ by adding a degree-3 vertex $v_n$, connecting it to a vertex $v_1$ of $G'$ of degree $d_1-1$, and inserting it on an arbitrary edge $v_iv_j$ of $G'$ with $v_1\notin\{v_i,v_j\}$ (replacing $v_iv_j$ with $v_nv_i$ and $v_nv_j$). Then $G$ is a realization of~$\dd$. Moreover, $G$ contains the edge-disjoint spanning trees $T_1:=T_1'\cup v_1v_n$ and $T_2:=(T_2'- v_iv_j) \cup \{v_nv_i, v_nv_j\}$.

        \smallskip\noindent\textit{Linear-Time Implementation:}
        This constructive algorithm uses an induction over~$n$.
        In the base case, we provided concrete TC-realizable graphs which can be handled in constant time.
        
        In the inductive step, we considered two cases.
        In Case 1, $\sum\dd>4(n-1)$ or $d_n\neq3$. 
        Here, the degree $d_n\geq 2$ is laid off.
        The resulting degree sequence~$\dd'$ allows for the desired property by the induction hypothesis, and a respective realization for~$\dd'$ can be computed in~$\Oh((n-1) + (m-d_n))$~time.
        Since we can lay off and reattach in $\Oh(d_n)$~time with our data structure, the total running time is $\Oh((n-1) + (m-d_n)) +\Oh(d_n)=\Oh(n+m)$~time since we can implement the induction by a linear loop.
        
        In Case 2, $\sum\dd=4(n-1)$ and $d_n=3$. Here, the entry $d_1$ is decreased by 1, and the entry $d_n=3$ is removed. This operation can be performed in constant time with our data structure. 
        The resulting degree sequence~$\dd'$ then allows for a realization $G'$ with two edge-disjoint spanning trees by the induction hypothesis and a respective realization for~$\dd'$ can be computed in~$\Oh((n-1) + (m-2))$~time.
        Then, to obtain a realization for $\dd$, one has to add a vertex, connect it to a vertex $v_1$ of degree $d_1-1$, and insert it on an edge not adjacent to $v_1$. 
        This operation takes only constant time, which leads to $\Oh(n+m)$ time total.
    \end{proof}

    For multigraphical sequences, we split the proof in two parts: we first show how to construct the desired realization for a multigraphical sequence $\dd$ with $\sum\dd=4(n-1)$ and $d_n\geq 2$ in \Cref{lem: multi constructing 2 edst in poly - equiality}, and then describe how this can be used to obtain a realization for every multigraphical sequence fulfilling $\sum\dd>4(n-1)$ and $d_n\geq 2$.
    
    \begin{lemma}\label{lem: multi constructing 2 edst in poly - equiality}
    
        Let $\dd=(d_1,\dots,d_n)$ be a multigraphical sequence with $\sum\dd = 4(n-1)$ and $d_n\geq 2$. 
        We can compute a realization~$G$ of $\dd$ and two disjoint spanning trees of~$G$ in $\Oh(n+m)$~time.
    \end{lemma}
    
    \begin{proof}
        To prove the statement, we show that for each $n$, for each multigraphical sequence~$\dd$ of length $n$ with $\sum\dd=4(n-1)$ and $d_n\geq 2$, we can construct a realization with two edge-disjoint spanning trees.
        We show this statement via induction over $n$.

        For the base case, consider $n\in\{2,3\}$. For $n=2$, $(2,2)$ is the only multigraphical sequence with $\sum\dd = 4(n-1)$ and $d_n\geq 2$. The multigraph $(\{v_1,v_2\},\{e_1=v_1v_2, e_2=v_1v_2\})$ is a realization of $(2,2)$ with two edge-disjoint spanning trees $T_1=e_1$ and $T_2=e_2$.
        For $n=3$, $(3,3,2)$ is the only multigraphical sequence with $\sum\dd = 4(n-1)$ and $d_n\geq 2$. The multigraph $(\{v_1,v_2,v_3\},\{e_1=v_1v_2, e_2=v_1v_2, e_3=v_1v_3, e_4=v_2v_3\})$ is a realization of $(3,3,2)$ with two edge-disjoint spanning trees $T_1=\{e_1,e_3\}$ and $T_2=\{e_2,e_4\}$.
        
        For the inductive step, let $n\geq 4$ and assume the statement holds for $n-1$. Let~$\dd$ be an arbitrary multigraphical sequence of length $n$ fulfilling $\sum\dd=4(n-1)$ and $d_n\geq 2$. We~describe how to construct a realization of $\dd$ with two edge-disjoint spanning trees.
        Observe that $d_1\leq\sum_{i=2}^n d_i$ and $\sum\dd$ is even, since $\dd$ is multigraphical.
        Furthermore, $d_n\in\{2,3\}$, as otherwise $\sum\dd\geq4n$.
        
        \medskip\noindent{\sffamily\bfseries\color{darkgray} Case 1:}\quad
            If $d_n=2$, consider the residual sequence $\dd' = (d_1', \dots, d_{n'}')$ of length $n'=n-1$ obtained from \dd by laying off the entry $d_n$ twice: reducing $d_n$ and $d_1$ by one and reordering; removing $d_n$ and reducing $d_i=\max(d_1-1,d_2)$ and reordering. 
            Due to \Cref{laying off multi}, $\dd'$ is multigraphical, and by $\sum\dd'=\sum\dd-4=4(n-2)$ and $d'_{n'}\geq d_{n-1}\geq 2$, $\dd'$ admits a realization $G'$ containing two edge-disjoint spanning trees $T_1'$ and $T_2'$.
                        Now consider the graph $G$ obtained from $G'$ by adding a vertex $v_n$ of degree $d_n=2$ and first connecting it to a vertex $v_i$ of degree $d_i=\max(d_1-1,d_2)$ and then connecting it to a vertex $v_1$ of degree $d_1-1$. Then $G$ is a realization of $\dd$. Moreover, $G$ contains the edge-disjoint spanning trees $T_1:=T'_1\cup v_1v_n$ and $T_2:=T'_2\cup v_iv_n$ (if~$v_1 = v_i$, these are then distinct multiedges).
    
        \medskip\noindent{\sffamily\bfseries\color{darkgray} Case 2:}\quad
            If~$d_n=3$, then by~$\sum\dd = 4(n-1)$, we get that~$d_1 \leq n-1$ and~$n\geq 4$.
            Thus, by~\Cref{realizable if dn geq 3}, $\dd$ is also a graphical sequence, for which~\Cref{lem: simple constructing 2 edst in poly} provides the desired construction. 
        \smallskip\noindent\textit{Linear-Time Implementation:}
        This constructive algorithm uses an induction over~$n$.
        In the base case of $n\leq 3$, we provide concrete realizations which can be handled in constant time each.
        
        In the inductive step, we considered two cases.
        In Case 1, $d_n=2$ and the entry $d_n$ is laid off.
        The resulting degree sequence~$\dd'$ then allows for a realization $G'$ with two edge-disjoint spanning trees by the induction hypothesis and a respective realization for~$\dd'$ can be computed in~$\Oh((n-1) + (m-2))$~time.
        Since we can lay off and reattach in $\Oh(d_n) = \Oh(1)$~time with our data structure, the total running time is then $\Oh(n+m)$~time since we can implement the induction by a linear loop.
        In Case 2, $d_n=3$ and the sequence is graphical. By \Cref{lem: simple constructing 2 edst in poly}, it allows for a realization with two edge-disjoint spanning trees by the induction hypothesis that can be computed in $\Oh(n+m)$ time.
    \end{proof}

    The construction for sequences with a degree sum larger than $4(n-1)$ follows directly from the laying-off process for multigraphical sequences (\Cref{laying off multi}), which reduces the degree sum by two while preserving the length of the sequence.
    
    \begin{lemma}\label{lem: multi constructing 2 edst in poly}
    
        Let $\dd=(d_1,\dots,d_n)$ be a multigraphical sequence with $\sum\dd\geq4(n-1)$ and $d_n\geq 2$. 
        We can compute a realization~$G$ of $\dd$ and two disjoint spanning trees of~$G$ in $\Oh(n+m)$~time.
    \end{lemma}
    
    \begin{proof}
        To prove the statement, we show that for each $n$, for each multigraphical sequence~$\dd$ of length $n$ with $\sum\dd\geq4(n-1)$ and $d_n\geq 2$, we can construct a realization with two edge-disjoint spanning trees.
        We show this statement via induction over $\sum\dd$.
        
        For the base case $\sum\dd=4(n-1)$ the claim follows from \Cref{lem: multi constructing 2 edst in poly - equiality}. 
        
        For the inductive step, let $\sum\dd> 4(n-1)$ and assume the statement holds for $\sum\dd-2$. Let~$\dd$ be an arbitrary multigraphical sequence with $\sum\dd>4(n-1)$ fulfilling  $d_n\geq 2$.
        Consider the residual sequence $\dd^{'} = (d_1', \dots, d_{n}')$ of length $n$ obtained from $\dd$ by reducing $d_2$ and $d_1$ by one. Note that $d_2>2$ since $d_1\leq\sum_{i=2}^n d_i$ and $\sum\dd>4(n-1)$.
                It follows that $\sum\dd'=\sum\dd-2\geq 4(n-1)$, $\dd'$ is multigraphical by \Cref{laying off multi}, and $d'_n\geq 2$. By the induction hypothesis, $\dd'$ admits a realization $G'$ containing two edge-disjoint spanning trees $T_1'$ and $T_2'$.
        
        Now consider the graph~$G$ obtained from~$G'$ by adding an edge between vertices $v_j$ of degree $d_j-1$ and $v_1$ of degree $d_1-1$ in~$G'$. Then, $G$ is a realization of~$\dd$ and contains the spanning trees~$T_1 := T'_1$ and~$T_2 := T'_2$.
        
        \smallskip\noindent\textit{Linear-Time Implementation:}
        This constructive algorithm uses an induction over~$\sum\dd$.
        In~the~base case of $\sum\dd=4(n-1)$, it follows from \Cref{lem: multi constructing 2 edst in poly - equiality} that the desired realization can be computed in $\Oh(n+m)$ time.
        
        In the inductive step, we decrease the entries $d_1$ and $d_2$ by 1 one each. The resulting degree sequence allows for a realization $G'$ with two edge-disjoint spanning trees that can be computed in $\Oh(n+(m-1))$ time by the induction hypothesis.
        Since we can add an edge between two vertices of degree $d_1-1$ and $d_2-1$ in $\Oh(1)$~time with our data structure, the total running time is $\Oh(n+m)$~time.
                                \end{proof}

\subsection{Realizations for spanning trees that share at most one edge and~\cp graph realizations}

We now go through the individual existential proofs of the previous sections and argue that all these proofs imply linear-time implementable constructions by using the described data structures to efficiently lay off degrees and reattach vertices.

\begin{lemma}

Let~$\dd$ be a (multi)graphical degree sequence.
If~$\dd$ is TC-realizable, then we can compute a realization of~$\dd$ in $\Oh(n+m)$~time.
\end{lemma}

\begin{proof}
We start with the proofs in~\Cref{sec:twotrees}.
For~$n \leq 2$, graphical sequences only have unique realizations, which can be clearly computed in constant time.
We thus focus on~$n > 2$.

Consider the construction of~\Cref{tc if deg 1}.
In this proof, the degree~$d_n = 1$ was laid off and the resulting degree sequence was shown to be graphical and admitting a realization with two edge-disjoint spanning trees.
Such a realization can be computed in $\Oh((n-1)+(m-1))$~time according to the previous algorithm.
Afterward the realization for~$\dd$ can be obtained by adding a degree-1 vertex in constant time.
In total this leads to a running time of $\Oh((n-1)+(m-1))$~time.

Note that in all further graphical cases, $m= \Theta(n)$.

Consider the construction of~\Cref{minus two dn eq 3}.
For constant size~$n$, we provided concrete realizations and for all other cases, the realization was obtained by adding a cycle of size~$x$ to a specific vertex in a realization of a specific graphical degree sequence of length~$n-x$ that allows for two edge-disjoint spanning trees and which has either constant size or fulfills the requirements of~\Cref{tc if deg 1}.
In both cases, the total running time is then $\Oh(x) + \Oh(n-x) = \Oh(n)$~time.

Consider the construction of~\Cref{minus two dn eq 2}.
Besides some constant size exceptions (which clearly can be handled in constant time each), the constructive algorithm used induction over~$n$.
In the inductive step, the degree~$d_n = 2$ was laid off.
The resulting degree sequence then either allowed for the desired property by the induction hypothesis or one of the previous lemmas and their respective algorithm.
Since we can lay off and reattach in $\Oh(d_n) = \Oh(1)$~time with our data structure, the total running time is then $\Oh(n)$~time since we can implement the induction by a linear loop.

Next, we consider the proofs of~\Cref{sec:c4}.

Consider the construction of~\Cref{minus four dn eq 3 d1 leq 4}.
The proof provides an induction over the length of~$n$ with base case~$n=8$.
In the inductive step, a new entry of value~$4$ is added.
To obtain a realization for the degree sequence of length~$n$, one takes a realization~$G'$ of the (unique) degree sequence of length~$n-1$ fulfilling the conditions of the induction.
In this graph~$G'$, a specific matching of size~$2$ is removed and a new vertex of degree~$4$ in introduced that becomes a neighbor of all four endpoints of the matching.
Note that this inductive step takes only constant time each if the specific matching is provided.
The latter is ensured by the induction hypothesis, since the graph has a constant maximum degree and a desired matching as described in the proof of~\Cref{minus four dn eq 3 d1 leq 4} can thus be found in constant time.
Clearly, this induction can then be implemented in $\Oh(n)$~time.

Consider the construction of~\Cref{minus four dn eq 3 d1 geq 5}.
It was shown that six entries of value~$3$ can be removed from~$\dd$ and~$d_1$ can be decreased by~$2$ to obtain a graphical degree sequence which admits a realization~$G'$ with two edge-disjoint spanning trees.
Such a realization can be computed by the initial algorithm in~$\Oh(n)$.
To obtain a realization for~$\dd$, the constant size gadget shown in~\Cref{fig:d1 greater 4 and dn equal 3} can be attached to an arbitrary vertex of degree~$d_1-2$ in~$G'$.
Since this gadget has constant size and constantly many edges attached to it, the \cp graph realization of~$\dd$ can be computed in $\Oh(n)$~time.

Consider the construction of~\Cref{minus four dn eq 2}.
Recall that similarly to~\Cref{minus two dn eq 2}, the constructive algorithm is again an induction over~$n$ by laying off~$d_n=2$ until reaching a sequence of constant length or fulfilling the conditions of~\Cref{minus four dn eq 3 d1 leq 4} or~\Cref{minus four dn eq 3 d1 geq 5}.
By the same arguments used for the linear time algorithm for~\Cref{minus two dn eq 2}, the constructive algorithm for~\Cref{minus four dn eq 2} can also be performed in $\Oh(n)$~time.

Finally, we consider the proofs of~\Cref{sec:multi} for multigraphical degree sequences.

Consider the construction of~\Cref{two spanning trees multi}.
For~$n\leq 2$, all multigraphical degree sequences only have a unique realization which can be found in $\Oh(n+m)$~time.
If~$n > 3$, (i)~the degree sequence was shown to be graphical (if~$d_n= 3$) for which the algorithms of~\Cref{sec:twotrees} imply the desired algorithm, or (ii)~the smallest degree can be laid off if~$d_n = 1$ or simply removed from the sequence if~$d_n=2$.
In both latter cases, the resulting sequence was shown to be multigraphical and allows for a realization with two edge-disjoint spanning trees.
By the initial algorithms for this case, such a realization can be computed in $\Oh((n-1) + (m-1))$~time and the lowest-degree vertex can be reattached in constant time afterwards, leading to $\Oh(n+m)$~time.

It remains to consider the construction of~\Cref{c4 sequences multi}.
If~$d_{n-2} \geq 3$, the degree sequence was shown to be graphical for which the algorithms of~\Cref{sec:c4} imply the desired algorithm.
Otherwise, the degree sequence had constant size or we could remove the entries~$d_{n-2}$, $d_{n-1}$, and~$d_{n}$ from~$\dd$ and reduce~$d_1$ by~$2$.
This operation can be done in constant time with our data structure.
Afterwards, the resulting sequence was shown to be multigraphical and allows for a realization~$G'$ with two edge-disjoint spanning trees.
By the initial algorithms for this case, such a realization can be computed in $\Oh((n-3) + (m-4))$~time.
Then, to obtain a TC-realizable graph for~$\dd$, one only has to add three vertices that together with an arbitrary vertex of degree~$d_1-2$ of~$G'$ will be an induced cycle of length 4 (namely the central cycle of the resulting~\cp multigraph).
Adding these 3 vertices and 4 edges takes only constant time, which leads to $\Oh(n+m)$~time in total.
\end{proof}

\subsection{Computing a TC labeling}

    Recall that during the construction of the TC-realizable graph~$G$, we always stored which edges belong to which of the two spanning trees $T_1$ and $T_2$, allowing us to extract both trees in linear time.
    By construction, the trees are edge-disjoint, share exactly one edge $e$, or share exactly two edges, in which case $G$ contains a central cycle $C$.
    In the latter two cases, the output was $e$ or $C$.
        
    We now describe how to efficiently label the edges of~$G$ using the extracted trees~$T_1$ and~$T_2$, the potential edge $e$, and the potential central cycle~$C$.
    The labeling procedures are based on the classical approach by Baker~and~Shostak~\cite{baker_gossips_1972}, which we show can be executed in linear time, given~$T_1$, $T_2$, $e$ and~$C$.
    For the sake of completeness, we recall these procedures.
    
    The labelings rely on the general idea that each vertex travels ``up'' along~$T_1$ to reach a central structure---a central vertex, an overlapping edge $e$, or an overlapping central cycle~$C$---and then ``down'' along~$T_2$ to reach every other vertex. This scheme is referred to as \emph{pivot labeling} in recent literature~\cite{casteigts_temporal_2021,christiann_inefficiently_2024}, where the central structure acts as a pivot (turning point).

    We begin by assigning a common root~$v_r$ to both~$T_1$ and~$T_2$.  
    \begin{itemize}
      \item If~$T_1$ and~$T_2$ are edge-disjoint, we choose an arbitrary vertex as the root.  
      \item If the trees share an edge or we return a central cycle, we contract the edge or cycle into a single root vertex.
    \end{itemize}
    
    Next, we label the edges of~$T_1$ in increasing order from the leaves to the root:
    \begin{itemize}
      \item Initialize a label counter~$t = 0$, and $X$ as the set $L$ of leaves of $T_1$.
      \item Assign the label $t+i$ to the edge~$\ell_ip_i$ where~$\ell_i\in X$ and $p_i$ is its parent in $T$. Set $t=t+L$ and $X=\{p_i:1\leq i\leq L\}$.
      \item       Assign the label $t+i$ to the unlabeled edge~$p_iq_i$ where~$p_i\in X$ and~$q_i$ is its parent in $T$.
      Set $t=t+\lvert X\rvert$ and $X=\{q_i:1\leq i\leq \lvert X\rvert\}$.
      \item Continue this process until all paths from the leaves have reached the root, \ie $X=\{v_r\}$.
    \end{itemize}
    In this way, the edges in~$T_1$ are labeled in increasing order as we move from the leaves to the root, and every vertex reaches the root by some time $t_r$.

    Next, we label the central structure. If the trees are edge-disjoint, no additional labeling is required. If they share a single edge~$e$, we assign it label~$t_r + 1$ to ensure that all vertices can traverse the edge. If there exists a shared central cycle~$C = (a, b, c, d)$, we assign label~$t_r + 1$ to the opposite edges~$ab$ and~$cd$, and label~$t_r + 2$ to the remaining two edges, that is, $\lambda(ab) = \lambda(cd) = t_r + 1$ and $\lambda(bc) = \lambda(da) = t_r + 2$.
    
    Finally, we label the edges of~$T_2$ in decreasing order from the root to the leaves:  
    \begin{itemize}
      \item Initialize a label counter~$t = t_r+3$, $X=\{v_r\}$, and $S_{v_r}$ as the set of children of $v_r$ in $T_2$.  
      \item Assign the label $t+i$ to the unlabeled edge $v_rs_i$ where $s_i\in X_{v_r}$. Set $t=t+\lvert S_{v_r} \rvert$, $X=S_{v_r}$, and $S_{s_i}$ as the set of children of each $s_i$ in $T_2$.
      \item For each vertex $x\in X$, assign the label $t+i$ to the unlabeled edge $xy_i$ where $y_i\in S_x$. Set $t=t+\lvert S_x\rvert$, $X=X-x\cup S_x$, and $S_{y_i}$ as the set of children of each $y_i$ in $T_2$.
            \item Repeat this process until all edges of~$T_2$ have been labeled, \ie $X=\emptyset$.
    \end{itemize}
    In this way, the edges in $T_1$ are labeled in increasing order as we move from the root to the leaves, and every vertex is reachable from the root with a temporal path starting at $t_r+3$.
    
    Given the trees $T_1$ and $T_2$ and the potential overlapping edge $e$ or central cycle $C$, this labeling procedure can be executed in time $\Oh(n+m)$ via BFS algorithms on both spanning trees.

We conclude the following algorithmic main result of our work.

\begin{theorem}

Let~$\dd$ be a (multi)graphical degree sequence.
We can decide in $\Oh(n+m)$~time whether~$\dd$ is TC-realizable.
If this is the case, then in $\Oh(n+m)$~time, we can compute a (multi)graph~$G$ which is a realization of~$\dd$ together with a TC labeling for~$G$. 
\end{theorem}

\section{Conclusion}
In this work, we presented complete characterizations of graphical and multigraphical degree sequences that admit temporally connected realizations under proper labelings.
All characterizations presented in this paper are constructive and lead to linear-time algorithms for constructing TC realizations (including the corresponding labelings), provided that the input satisfies the corresponding conditions.

As explained in the introduction, the fact that we only considered proper labelings is not a limitation, as the remaining cases involving non-proper labelings are either trivial or covered by our techniques.
Thus, our results cover directly or indirectly all the possible combinations of settings among proper/non-proper, simple/non-simple, and strict/non-strict. 

Several natural questions could follow from this work.
For instance, one may ask to enumerate all TC-realizable graphs for a given degree sequence, or ask for realizations with additional constraints such as bounded lifetime, bounded diameter, or the existence of temporal spanners of a certain size.
Another promising avenue is to optimize over these criteria while searching for a realization.
Finally, the realizability question studied in this paper could be generalized to temporal~$k$-connectivity or to directed temporal graphs.

\bibliographystyle{abbrv}
\bibliography{references,zotero-arnimi}
\end{document}